\documentclass[11pt,onecolumn,draftcls]{IEEEtran}
\usepackage{ifpdf}
\usepackage{cite}
\usepackage{graphicx}
\usepackage[cmex10]{amsmath}
\usepackage{algorithmic}
\usepackage{array}
\usepackage{mdwmath}
\usepackage{mdwtab}
\usepackage{eqparbox}
\usepackage[tight,footnotesize]{subfigure}
\usepackage{stfloats}
\usepackage{amssymb}
\usepackage{bm}
\usepackage{overpic}
\usepackage{color}
\usepackage{dsfont}
\usepackage{amsthm}
\usepackage{cases}

\newtheorem{theorem}{Theorem}
\newtheorem{lemma}{Lemma}

\newtheorem{definition}{Definition}

\newtheorem{assumption}{Assumption}

\def\E{\mathbb{E}}
\def\mb{\mathbb}
\def\defeq{\triangleq}
\def\mc{\mathcal}
\def\col{\mathrm{col}}
\def\var{\mathrm{var}}
\def\diag{\mathrm{diag}}

\allowdisplaybreaks

\title{ Distributed Pareto Optimization via Diffusion Strategies}

\begin{document}

\author{Jianshu~Chen,~\IEEEmembership{Student~Member,~IEEE,}%
        ~and~Ali~H.~Sayed,~\IEEEmembership{Fellow,~IEEE}
\thanks{The authors are with Department of Electrical Engineering,
University of California, Los Angeles, CA 90095. Email: \{jshchen, sayed\}@ee.ucla.edu.
}
\thanks{This work was supported in part by NSF grants CCF-1011918 and CCF-0942936.
A preliminary short version of this work is reported in the conference publication%
\cite{chen2012ssp}.}
}


\maketitle

\begin{abstract}
We consider solving multi-objective optimization problems in a distributed manner
by a network of cooperating and learning agents. The problem is equivalent to optimizing a global cost
that is the sum of individual components. The optimizers of the individual components do not necessarily coincide and the network therefore needs to seek Pareto optimal solutions. We develop a distributed solution that relies on a general class of adaptive diffusion strategies. We show how the diffusion process can be represented as the cascade composition of three operators: two combination operators and a gradient descent operator. Using the Banach fixed-point theorem, we establish the existence of a unique fixed point for the composite cascade. We then study how close each agent converges towards this fixed point, and also examine how close the Pareto solution is to the fixed point. We perform a detailed mean-square error analysis and establish that all agents are able to converge to the same Pareto optimal solution within a sufficiently small mean-square-error (MSE) bound even for constant step-sizes.
We illustrate one application of the theory to collaborative decision making in finance by a network of agents. 
\end{abstract}
\section{Introduction}
\label{Sec:Intro}
We consider solving a multi-objective optimization
problem in a distributed manner over a network of $N$ cooperative learners (see Fig. \ref{Fig:Fig_Network}).
Each agent $k$ is associated with an individual cost function $J_k^o(w)$; and each of these
costs may not be minimized at the
same vector $w^o$. As such, we need to seek a solution that is ``optimal'' in some sense for the entire
network.
In these cases, a general concept of
optimality known as \emph{Pareto optimality} is useful to characterize
how good a solution is. A solution $w^o$ is said to be Pareto optimal if there
does not exist another vector $w$ that is able to improve (i.e., reduce) any particular
cost, say, $J_k^o(w)$, without
degrading (increasing) some of the other costs $\{J_l^o(w)\}_{l \neq k}$.
To illustrate the idea of Pareto optimality, let
	\begin{align}
		\label{Equ:ProblemFormulation:O_FeasibleObjectiveValues}
		\mc{O}	&\defeq	\{ (J_1^o(w),\ldots,J_N^o(w)): \; w \in \mathbb{W} \}  \subseteq \mathbb{R}^N
	\end{align}
denote the set of achievable cost values, where $\mathbb{W}$ denotes the feasible set.
Each point $P \in \mc{O}$ represents attained values for the cost functions $\{J_l^o(w)\}$
at a certain $w \in \mathbb{W}$.
Let us consider the two-node case ($N=2$) shown in Fig. \ref{Fig:Fig_ParetoOptimality},
where the shaded areas represent the set $\mc{O}$ for two situations
of interest.
	\begin{figure}[h]
        \centering
        \includegraphics[width=0.45\textwidth]{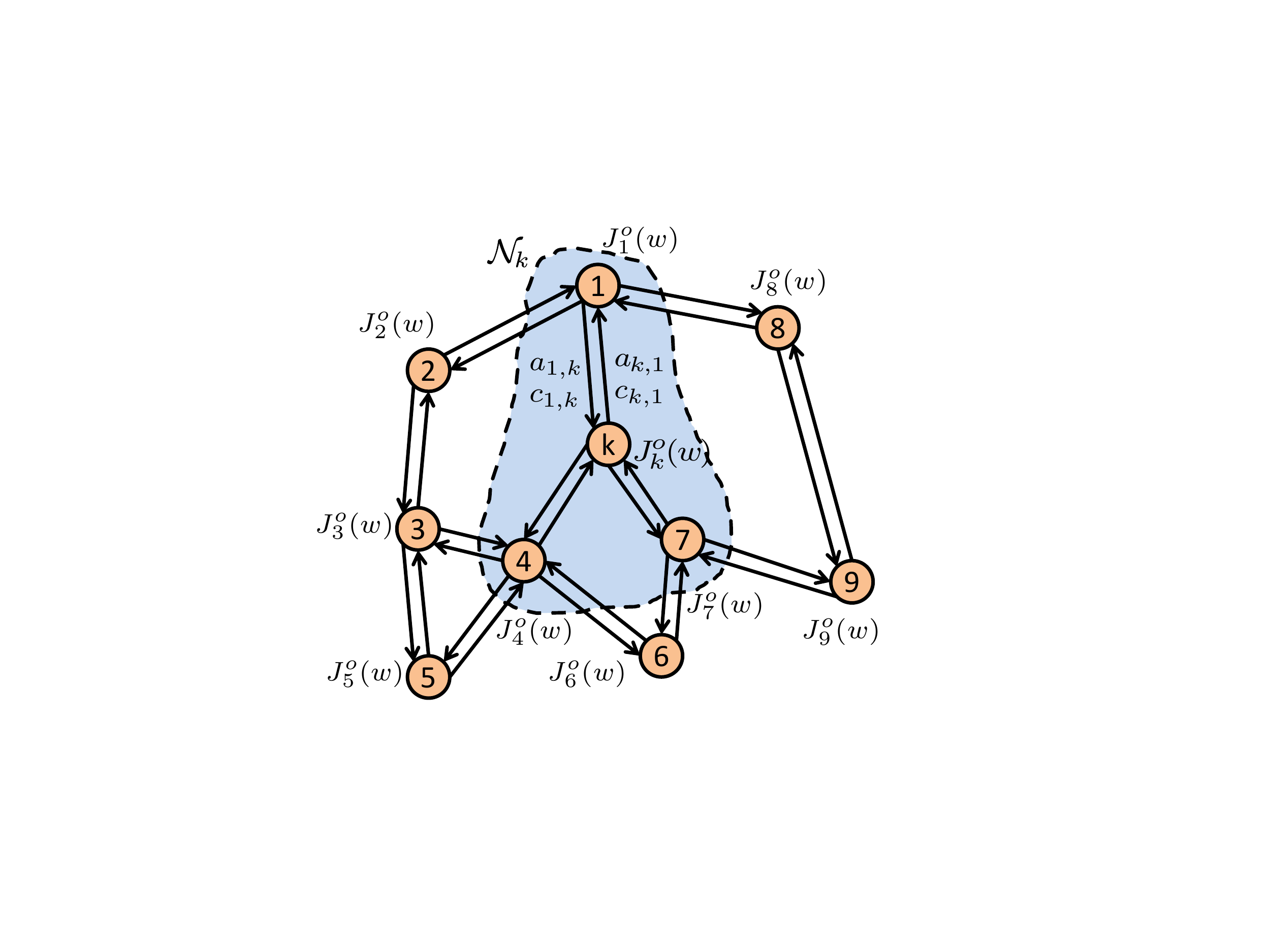}
        \caption{A network of $N$ cooperating agents; a cost function $J_k^o(w)$ is associated with each node $k$. The set of neighbors of node $k$ is denoted by ${\cal N}_k$ (including node $k$ itself); this set consists of all nodes with which node $k$ can share information. }
        \label{Fig:Fig_Network}
        \vspace{-0.5\baselineskip}
    \end{figure}
In Fig. \ref{Fig:Fig_ParetoOptimality:Single}, both $J_1^o(w)$ and $J_2^o(w)$ achieve their minima
at the same point $P = (J_1^o(w^o), J_2^o(w^o))$, where $w^o$ is the common minimizer.
In comparison, in Fig. \ref{Fig:Fig_ParetoOptimality:Multiple}, $J_1^o(w)$ attains its minimum at point $P_1$,
while $J_2^o(w)$ attains its minimum at point $P_2$, so that they do not have a common minimizer.
Instead, all the points on
the heavy red curve between points $P_1$ and $P_2$ are Pareto optimal solutions. For example,
starting at point $A$ on the curve, if we want to reduce the value of $J_1^o(w)$
without increasing the value of $J_2^o(w)$, then we will need to move out of
the achievable set $\mc{O}$. The alternative choice that would keep us on the curve is to move
to another Pareto optimal point $B$, which would however increase the value of $J_2^o(w)$.
In other words, we need to trade the value of $J_2^o(w)$ for $J_1^o(w)$. For this reason, the
curve from $P_1$ to $P_2$ is called the optimal tradeoff curve (or optimal tradeoff surface if $N>2$)
\cite[p.183]{boyd2004convex}.

	\begin{figure}[t]
		\centerline{
			\subfigure[]{
				\includegraphics[width=0.35\textwidth]{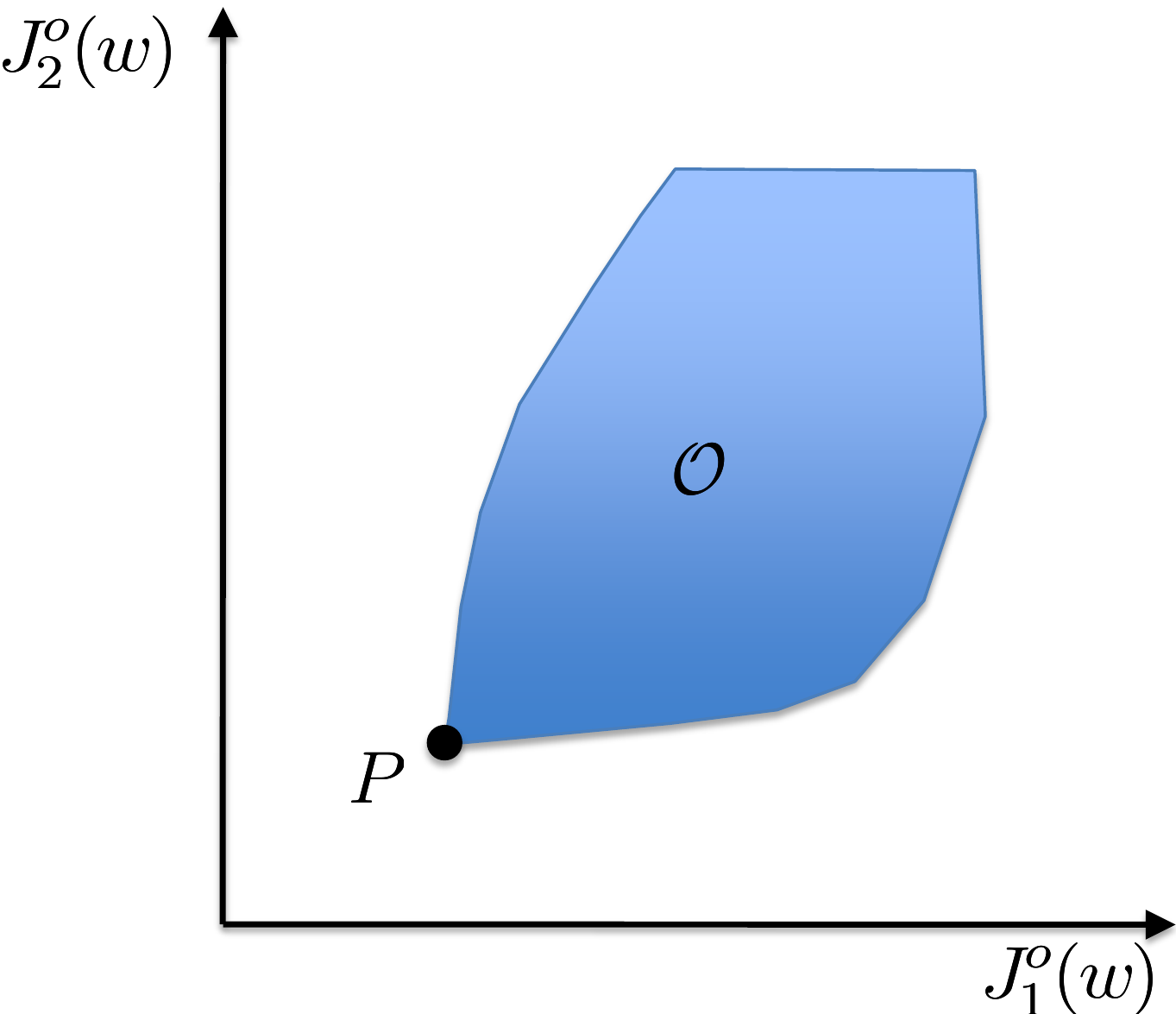}
				\label{Fig:Fig_ParetoOptimality:Single}
			}
			\subfigure[]{
				\includegraphics[width=0.35\textwidth]{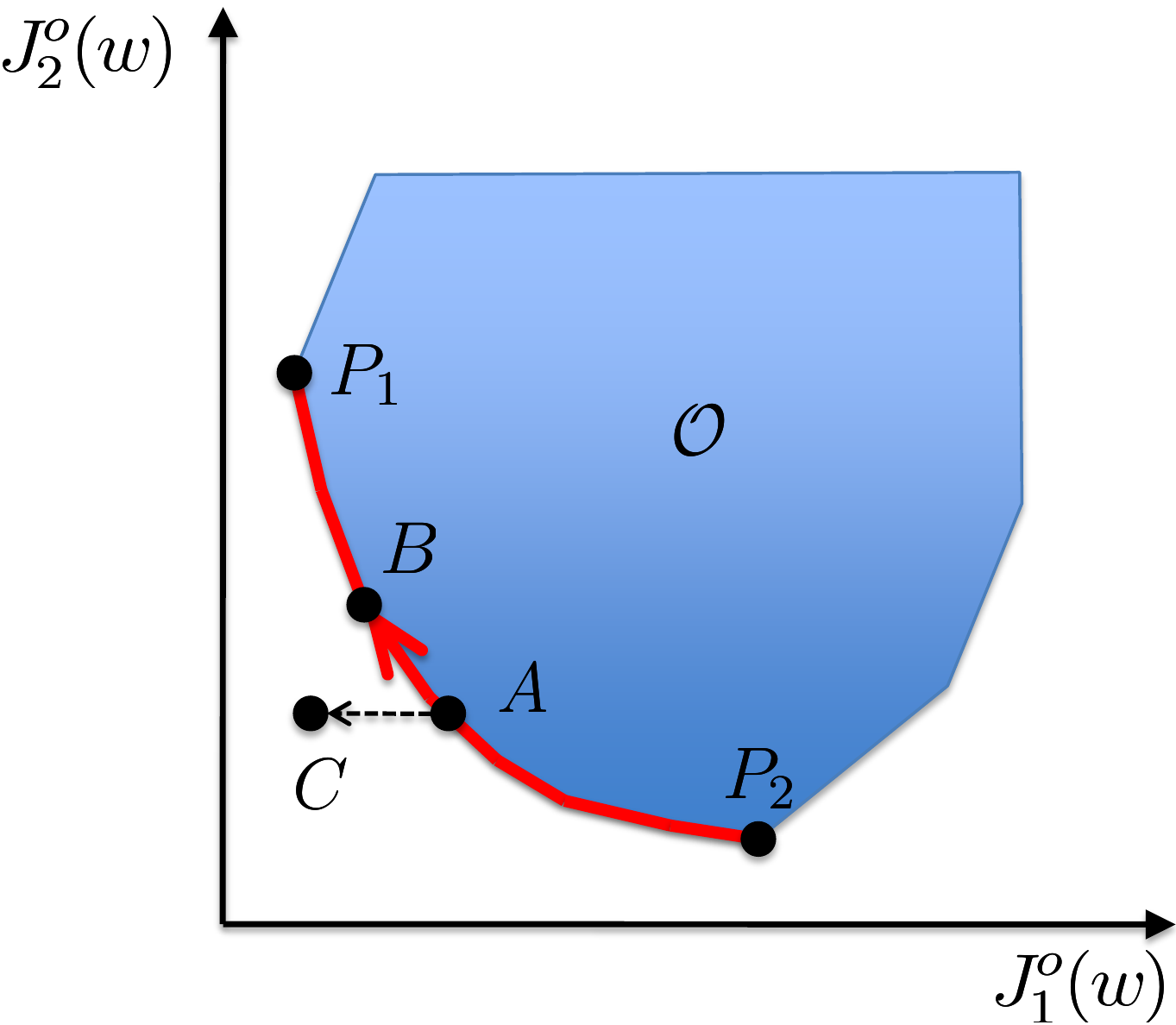}
				\label{Fig:Fig_ParetoOptimality:Multiple}
			}
		}
		\caption{Optimal and Pareto optimal points for the case $N=2$:
        (Left)  $P$ denotes the optimal point where both cost functions
        are minimized simultaneously and (Right) Pareto optimal points
        lie on the red boundary curve.}
		\label{Fig:Fig_ParetoOptimality}
		\vspace{-\baselineskip}
	\end{figure}

To solve for Pareto optimal solutions, a useful scalarization technique is often used
to form an aggregate cost function that is the weighted sum of the component costs as follows:
	\begin{align}
		\label{Equ:ProblemFormulation:J_glob_scalarization}
		J^{\mathrm{glob}}(w)		=	\sum_{l=1}^N \pi_l J_l^o(w)
	\end{align}
where $\pi_l$ is a positive weight attached with the $l$th cost.
It was shown in \cite[pp.178--180]{boyd2004convex} that the minimizer of
\eqref{Equ:ProblemFormulation:J_glob_scalarization} is Pareto optimal
for the multi-objective optimization problem. Moreover, by varying the values
of $\{\pi_l\}$, we are able to get different Pareto optimal points
on the tradeoff curve.
Observing that we can always define a new cost $J_l(w)$ by 
incorporating the weighting scalar $\pi_l$, 
    \begin{align}
        \label{Equ:ProblemFormulation:J_l_newcost}
        J_l(w) \defeq \pi_l J_l^o(w)
    \end{align}
it is sufficient for our future discussions to focus on 
aggregate costs of the following form:
	\begin{align}
		\label{Equ:ProblemFormulation:J_glob_original}
		\boxed{
			J^{\mathrm{glob}}(w)	=	\sum_{l=1}^N J_l(w)			
		}
	\end{align}
If desired, we can also add constraints to problem \eqref{Equ:ProblemFormulation:J_glob_original}. For example, suppose there is additionally some constraint of the form $p_k^T w < b_k$ at node $k$, where
$p_k$ is $M \times 1$ and $b_k$ is a scalar. Then, we can consider using barrier
functions to convert the constrained optimization problem to an
unconstrained problem \cite{boyd2004convex,towfic2012icc}.  For example, we can redefine each
cost $J_k(w)$ to be $J_k(w) \leftarrow J_k(w)+\phi(p_k^Tw-b_k)$, where $\phi(x)$
is a barrier function that penalizes values of $w$ that violate the constraint.
Therefore, without loss of generality, we shall assume $\mathbb{W}=\mathbb{R}^M$ and only consider
unconstrained optimization problems.
Moreover, we shall assume the $\{J_l^o(w)\}$ are differentiable and,
{\color{black} for each given set of positive weights $\{\pi_l\}$,} the cost
$J^{\mathrm{glob}}(w)$ in \eqref{Equ:ProblemFormulation:J_glob_scalarization}
or \eqref{Equ:ProblemFormulation:J_glob_original}
is strongly convex so that the minimizer $w^o$ is unique  \cite{poliak1987introduction}.
Note that the new cost $J_l(w)$ in \eqref{Equ:ProblemFormulation:J_l_newcost} depends on
$\pi_l$ so that the $w^o$ that minimizes $J^{\mathrm{glob}}(w)$ in
\eqref{Equ:ProblemFormulation:J_glob_original} also depends on
$\{\pi_l\}$.

One of the most studied approaches to the distributed solution of such optimization problems is the incremental approach --- see, e.g.,
\cite{bertsekas1997parallel,bertsekas1997new,nedic2001incremental,
nedic2000convergence,polyak1973pseudogradient,rabbat2005quantized,lopes2007incremental,li2008new}.
In this approach, a cyclic path is defined over the nodes and data are processed in a cyclic manner through the network until optimization is achieved. However, determining a cyclic path that covers all nodes is generally an NP-hard
problem\cite{karp1972reducibility} and, in addition, cyclic trajectories are vulnerable to link and node failures. Another useful distributed optimization approach relies on the use of consensus strategies\cite{bertsekas1997parallel,tsitsiklis1986distributed,barbarossa2007bio,nedic2009distributed,
kar2011converegence,dimakis2010gossip,srivastava2011distributed,khan2011networked}. In this approach, vanishing step-size sequences are used to ensure that agents reach consensus and agree about the optimizer in steady-state.
However, in time-varying environments, diminishing step-sizes prevent the network from continuous learning;
when the step-sizes die out, the network stops learning.

In \cite{chen2011TSPdiffopt}, we generalized our earlier work on adaptation
and learning over networks \cite{lopes2008diffusion, Cattivelli10} and developed diffusion strategies that enable the decentralized optimization of global cost functions of the form \eqref{Equ:ProblemFormulation:J_glob_original}. In
the diffusion approach, information is processed locally at the nodes and then diffused through a real-time sharing mechanism. In this manner, the approach is scalable, robust to node and link failures, and avoids the need for cyclic trajectories. In addition, compared to the aforementioned consensus solutions (such as those in
\cite{nedic2009distributed,srivastava2011distributed,ram2010distributed}),
the diffusion strategies we consider here employ \emph{constant} (rather than vanishing)
step-sizes in order to endow the resulting networks with \emph{continuous} learning and
tracking abilities.
By keeping the step-sizes constant, the agents are able to track drifts in the underlying costs and in the location of the Pareto optimal solutions. One of the main challenges in the ensuing analysis
becomes that of showing that the agents are still able to approach the Pareto optimal solution even with constant step-sizes; in this way, the resulting diffusion strategies are able to combine the two useful properties of optimality and adaptation.

In \cite{chen2011TSPdiffopt},
we focused on the important case where all costs $\{J_l(w)\}$ share the {\em same} optimal solution $w^o$ (as was the case with Fig. \ref{Fig:Fig_ParetoOptimality:Single});
this situation arises when the agents in the network have a common objective and they cooperate to solve the problem of mutual interest in a distributed manner. Examples abound in biological networks where agents work together, for example, to locate food sources or evade predators\cite{tu2011fish}, and in
collaborative spectrum sensing\cite{di2011biotsp},
system identification\cite{chouvardas2011adaptive}, and learning
applications\cite{theodoridis2011adaptive}.
In this paper, we develop the necessary theory to show that the same
diffusion approach (described by \eqref{Equ:DiffusionAdaptation:ATC0_NE}--\eqref{Equ:DiffusionAdaptation:CTA0_NE} below) can be used to
solve the more challenging \emph{multi-objective} optimization problem,
where the agents need to converge instead to a Pareto optimal solution. Such situations
are common
in the context of multi-agent decision making (see, e.g.,
\cite{towfic2012icc} and also Sec. \ref{Sec:Simulation} where we discuss one application in the context of collaborative decision in finance).
To study this more demanding scenario, we first show that the proposed diffusion process can be represented as the cascade composition of three operators: two combination (aggregation) operators and one gradient-descent operator. Using the Banach fixed-point theorem\cite[pp.299--303]{kreyszig1989introductory}, we establish the existence of a unique fixed point for the composite cascade. We then study how close each agent
in the network converges towards this fixed point, and also examine how
close the Pareto solution is to the fixed point. We perform a detailed mean-square error analysis and establish that all agents are able to converge to the \emph{same} Pareto optimal solution within a sufficiently small mean-square-error (MSE) bound. We illustrate the results by considering an example involving collaborative decision in financial applications. 

\vspace{1em}

\noindent
{\bf Notation}. Throughout the paper, all vectors are column vectors.
We use boldface letters to denote random quantities (such as $\bm{u}_{k,i}$)
and regular font to denote their realizations or deterministic variables (such as $u_{k,i}$).
We use $\mathrm{diag}\{x_1,\ldots,x_N\}$ to denote a (block) diagonal matrix consisting of diagonal entries (blocks)
$x_1,\ldots,x_N$,
and use $\mathrm{col}\{x_1,\ldots,x_N\}$ to denote a column vector formed by stacking $x_1,\ldots,x_N$
on top of each other. The notation $x \preceq y$ means each entry of the
vector $x$ is less than or equal to the corresponding entry of the vector $y$.


\section{Diffusion Adaptation Strategies}
\label{Sec:Diffusion}

In \cite{chen2011TSPdiffopt}, we motivated and derived diffusion strategies for distributed optimization, 
which are captured by the following general description: 
{\setlength{\jot}{-5pt}
	\begin{align}	
		\label{Equ:PerformanceAnalysis:Diffusion_General}
			{\phi}_{k,i-1}	&=	 \displaystyle\sum_{l = 1 }^N a_{1,lk} {w}_{l,i-1}		\\
		\label{Equ:PerformanceAnalysis:Diffusion_General1}
			{\psi}_{k,i}		&=									
								\displaystyle{\phi}_{k,i-1}
								-
								\mu_k \sum_{l =1}^N
								c_{lk}
									{\nabla}_w J_l({\phi}_{k,i-1})
								\\
		\label{Equ:PerformanceAnalysis:Diffusion_General2}
			{w}_{k,i}		&=	 \displaystyle\sum_{l = 1}^N a_{2,lk} {\psi}_{l,i}
	\end{align}
}%
where ${w}_{k,i}$ is the local estimate for $w^o$ at node $k$ and time $i$,
$\mu_k$  is the step-size parameter used by node $k$,
and $\{\phi_{k,i-1},\psi_{k,i}\}$ are intermediate  estimates for $w^o$.
Moreover, $\nabla_w J_{l}(\cdot)$ is the (column) gradient vector of $J_l(\cdot)$ relative to $w$. 
The non-negative coefficients $\{a_{1,lk}\}$, $\{c_{lk}\}$, and $\{a_{2,lk}\}$ are
the $(l,k)$-th entries of matrices $A_1$, $C$, and $A_2$, respectively,
and they are required to satisfy:
	\begin{equation}
		\label{Equ:PerformanceAnalysis:ConditionCombinationWeights}
			\left\{
			\begin{split}
				 &A_1^T \mathds{1} = \mathds{1}, \;
				A_2^T \mathds{1} = \mathds{1}, \;
				C \mathds{1} = \mathds{1}
				\\
				&a_{1,lk}=0,~a_{2,lk}=0,~c_{lk}=0 \mathrm{~if~} l \notin \mathcal{N}_{k}
			\end{split}
			\right.
	\end{equation}
where $\mathds{1}$ denotes a vector with all entries equal to one. Note from
\eqref{Equ:PerformanceAnalysis:ConditionCombinationWeights}
that the combination coefficients
$\{a_{1,lk},a_{2,lk},c_{lk}\}$ are nonzero only for those $l \in \mathcal{N}_k$.
Therefore, the sums in \eqref{Equ:PerformanceAnalysis:Diffusion_General}--%
\eqref{Equ:PerformanceAnalysis:Diffusion_General2} are confined within the
neighborhood of node $k$.
Condition \eqref{Equ:PerformanceAnalysis:ConditionCombinationWeights} requires
the combination matrices $\{A_1,A_2\}$ to be left-stochastic,
while $C$ is right-stochastic. We therefore note that each node $k$
first aggregates the existing estimates from its neighbors through
\eqref{Equ:PerformanceAnalysis:Diffusion_General} and generates the intermediate estimate
$\phi_{k,i-1}$. Then, node $k$ aggregates gradient information from its neighborhood and updates $\phi_{k,i-1}$ to $\phi_{k,i}$ through \eqref{Equ:PerformanceAnalysis:Diffusion_General1}.
All other nodes in the network are performing these same steps simultaneously.
Finally, node $k$ aggregates the estimates $\{\phi_{l,i}\}$ through step
\eqref{Equ:PerformanceAnalysis:Diffusion_General2} to update its weight estimate to $w_{k,i}$.

Algorithm \eqref{Equ:PerformanceAnalysis:Diffusion_General}--%
\eqref{Equ:PerformanceAnalysis:Diffusion_General2} can be simplified to several special cases for different choices of the matrices  $\{A_1,A_2,C\}$. 
For example,
the choice $A_1 = I$, $A_2=A$ and $C=I$ reduces to
the adapt-then-combine (ATC) strategy that has no exchange of gradient information
\cite{lopes2008diffusion,chen2011TSPdiffopt,Cattivelli10,takahashi2010link}:
	\begin{align}
        \boxed{
				\label{Equ:DiffusionAdaptation:ATC0_NE}
				\begin{array}{l}
					\psi_{k,i}	=	 \displaystyle
								w_{k,i-1} - \mu_k
								\nabla_w J_k(w_{k,i-1})	\\
					w_{k,i}	=	\displaystyle \sum_{l \in \mathcal{N}_k} a_{lk} \psi_{l,i}
				\end{array}
			}\;\;(\mathrm{ATC},\;C=I)
    \end{align}
while the choice $A_1=A$, $A_2=I$ and $C=I$  reduces to the
combine-then-adapt (CTA) strategy, 
where the order of the combination and adaptation steps are reversed relative to
\eqref{Equ:DiffusionAdaptation:ATC0_NE}
\cite{lopes2008diffusion,Cattivelli10,takahashi2010link}:
	\begin{align}
        \boxed{
				\label{Equ:DiffusionAdaptation:CTA0_NE}
				\begin{array}{l}
					\psi_{k,i\!-\!1}	=	\displaystyle \sum_{l \in \mathcal{N}_k} a_{lk} w_{l,i-1}
                            		 \\
					w_{k,i}	=	 \displaystyle\psi_{k,i-1}
								- \mu_k \nabla_w J_k(\psi_{k,i-1})
				\end{array}
			}\;\;(\mathrm{CTA},\;C=I)
    \end{align}
Furthermore, if in the CTA implementation \eqref{Equ:DiffusionAdaptation:CTA0_NE}
we enforce $A$ to be doubly stochastic,
replace $\nabla_w J_k(\cdot)$ by a {\color{black}subgradient}, and
use a time-decaying step-size parameter ($\mu_k(i)\rightarrow 0$), then
we obtain the unconstrained
version used by \cite{ram2010distributed}.
In the sequel, we continue with the general recursions
\eqref{Equ:PerformanceAnalysis:Diffusion_General}--\eqref{Equ:PerformanceAnalysis:Diffusion_General2},
which allow us to examine the convergence properties of several algorithms in a unified manner.
The challenge we encounter now is to show that this same class of algorithms
can still optimize the cost \eqref{Equ:ProblemFormulation:J_glob_original}
in a  distributed manner when the individual costs $\{J_l(w)\}$ do not necessarily have the same minimizer. 
This is actually a demanding task, as the analysis in the coming sections reveals, 
and we need to introduce novel analysis techniques to be able to handle this general case.

\section{Performance Analysis}
\label{Sec:PerformanceAnalysis}
\subsection{Modeling Assumptions}
In most situations in practice, the true gradient vectors needed in
\eqref{Equ:PerformanceAnalysis:Diffusion_General1}
are not available. Instead, perturbed versions are available, which we model as
	\begin{align}
		\label{Equ:Diffusion:NoisyGradient}
		\widehat{\nabla_w J_l}(\bm{w}) = \nabla_w J_l(\bm{w}) + \bm{v}_{l,i}(\bm{w})
	\end{align}
where the random noise term, $\bm{v}_{l,i}(\bm{w})$, may depend on $\bm{w}$
and will be required to satisfy certain conditions given by
\eqref{Assumption:GradientNoise:ZeroMean_Uncorrelated}--\eqref{Assumption:GradientNoise:Norm}.
We refer to the perturbation in \eqref{Equ:Diffusion:NoisyGradient}
as gradient noise.
Using \eqref{Equ:Diffusion:NoisyGradient}, the diffusion algorithm \eqref{Equ:PerformanceAnalysis:Diffusion_General}--%
\eqref{Equ:PerformanceAnalysis:Diffusion_General2} becomes the following, where
we are using boldface letters for various quantities to highlight the fact that they are now stochastic in nature due to the randomness in the noise component:
{\setlength{\jot}{-5pt}
    \begin{align}	
		\label{Equ:PerformanceAnalysis:Diffusion_General_Noisy}
			\bm{\phi}_{k,i-1}	&=	 \displaystyle\sum_{l = 1 }^N a_{1,lk} \bm{w}_{l,i-1}		\\
		\label{Equ:PerformanceAnalysis:Diffusion_General1_Noisy}
			\bm{\psi}_{k,i}		&=									
								\displaystyle\bm{\phi}_{k,i-1}
								-
								\mu_k \sum_{l =1}^N
								c_{lk}
								\big[{\nabla}_w J_l(\bm{\phi}_{k,i-1}) + \bm{v}_{l,i}(\bm{\phi}_{k,i-1})\big]
								\\
		\label{Equ:PerformanceAnalysis:Diffusion_General2_Noisy}
			\bm{w}_{k,i}			&=	 \displaystyle\sum_{l = 1}^N a_{2,lk} \bm{\psi}_{l,i}
	\end{align}
}%
Using \eqref{Equ:PerformanceAnalysis:Diffusion_General_Noisy}--%
\eqref{Equ:PerformanceAnalysis:Diffusion_General2_Noisy}, we now proceed to
examine the mean-square performance of the diffusion strategies. 
Specifically,
in the sequel, we study: (i) how fast and (ii) how close the estimator $\bm{w}_{k,i}$ at each node $k$
approaches the Pareto-optimal
solution $w^o$ in the mean-square-error sense. We establish the convergence of
all nodes towards the same Pareto-optimal solution within a small MSE bound. 
Since we are dealing with individual costs that may not have a common minimizer, the approach we employ to examine the convergence properties of the diffusion
strategy is fundamentally different from \cite{chen2011TSPdiffopt}; we follow a
system-theoretic approach and call upon the fixed-point theorem for contractive mappings\cite[pp.299--303]{kreyszig1989introductory}.

To proceed with the analysis, we introduce the following assumptions on the cost functions and
gradient noise.
{\color{black}
As explained in
\cite{chen2011TSPdiffopt}, these conditions are weaker than similar conditions in the literature of distributed optimization; in this way,
our convergence and performance results hold under more relaxed conditions than usually
considered in the literature.}
	\begin{assumption}[Bounded Hessian]
		\label{Assumption:Hessian}
		Each component cost function $J_l(w)$ has a bounded Hessian matrix, i.e.,
		there exist
		nonnegative real numbers $\lambda_{l,\min}$ and $\lambda_{l,\max}$ such that,
		for each $k=1,\ldots,N$:
			\begin{align}
				\label{Equ:Assumption:StrongConvexity}
				&\lambda_{l,\min} I_M	\le	\nabla_w^2 J_l(w)		\le	\lambda_{l,\max} I_M
			\end{align}
		with $\sum_{l=1}^N c_{lk} \lambda_{l,\min} >0$.
        \hfill\qed
	\end{assumption}
	\begin{assumption}[Gradient noise]
		\label{Assumption:GradientNoise}
		There exist $\alpha \ge 0$ and $\sigma_{v}^2 \ge 0$ such that, for all
        $\bm{w} \in \mathcal{F}_{i-1}$:
			\begin{align}
				\label{Assumption:GradientNoise:ZeroMean_Uncorrelated}
				&\mathbb{E}\left\{\bm{v}_{l,i}(\bm{w}) \;|\; \mathcal{F}_{i-1} \right\} = 0	 \\
				\label{Assumption:GradientNoise:Norm}
				&\mathbb{E}\left\{ \|\bm{v}_{l,i}(\bm{w}) \|^2 \right\}
						\le \alpha \cdot \E \|\nabla_w J_l(\bm{w})\|^2 + \sigma_{v}^2
			\end{align}
		for all $i,l$, where $\mathcal{F}_{i-1}$ denotes the past history of
        estimators $\{\bm{w}_{k,j}\}$ for $j \le i-1$ and all $k$.
        \hfill\qed
	\end{assumption}
If we choose $C=I$, then Assumption \ref{Assumption:Hessian} implies that the cost functions $\{J_l(w)\}$ are strongly convex%
\footnote{A differentiable function $f(x)$ on $\mathbb{R}^n$ is said to be strongly convex
if there exists a $\lambda_{\min}>0$ such that
$f(x+y) \ge f(x) + y^T\nabla f(x)+\lambda_{\min} \|y\|^2/2$ for any $x,y \in \mathbb{R}^n$.
And if $f(x)$ is twice-differentiable,
this is also equivalent to $\nabla^2 f(x) \ge \lambda_{\min} I$\cite[pp.9-10]{poliak1987introduction}.
Strong convexity implies that the function $f(x)$ can be lower bounded by some quadratic function.
}.
This condition can be guaranteed by adding small regularization terms. For example,
we can convert a non-strongly convex function $J_l'(w)$ to a strongly convex one by 
redefining $J_l(w)$ as $J_l(w)  \leftarrow   J_l(w) + \epsilon \|w\|^2$,
where $\epsilon>0$ is a small regularization factor.
We further note that, assumption  \eqref{Assumption:GradientNoise:Norm} is a mix of the ``relative random noise''
and ``absolute random noise'' model usually assumed in stochastic approximation \cite{poliak1987introduction}.
Condition \eqref{Assumption:GradientNoise:Norm}
implies that the gradient noise grows  when the estimate is away
from the optimum (large gradient). Condition \eqref{Assumption:GradientNoise:Norm}
also states that even when the
gradient vector is zero, there is still some residual
noise variance $\sigma_v^2$.

\subsection{Diffusion Adaptation Operators}
To analyze the performance of the diffusion adaptation strategies, we first
represent the mappings performed by
\eqref{Equ:PerformanceAnalysis:Diffusion_General_Noisy}--%
\eqref{Equ:PerformanceAnalysis:Diffusion_General2_Noisy}
in terms of useful operators.

\begin{definition}[Combination Operator]
    \label{Def:CombinationOperator}
    Suppose $x=\col\{x_1,\ldots,x_N\}$ is an arbitrary $N \times 1$ block column vector that is formed
    by stacking $M \times 1$ vectors $x_1,\ldots, x_N$ on top of each other.
    The combination operator $T_A: \mb{R}^{MN} \rightarrow \mb{R}^{MN}$
    is defined as the linear mapping:
        \begin{align}
            \label{Equ:def:T_A}
            T_A(x)  \defeq   (A^T \otimes I_M) \; x
        \end{align}
    where $A$ is an $N \times N$ left stochastic matrix, and
    $\otimes$ denotes the Kronecker product operation. \hfill\qed
\end{definition}

\begin{definition}[Gradient-Descent Operator]
    \label{Def:GradientOperator}
    Consider the same $N\times 1$ block column vector $x$. Then,
    the gradient-descent operator $T_G: \mb{R}^{MN} \rightarrow \mb{R}^{MN}$
    is the nonlinear mapping defined by:
        \begin{align}
            \label{Equ:Def:T_G}
            T_G(x)  \defeq  
							\begin{bmatrix}
                                x_1 - \mu_1 \sum_{l=1}^N c_{l1} \nabla_w J_l(x_1)\\
                                    \vdots\\
                                x_N - \mu_N \sum_{l=1}^N c_{lN} \nabla_w J_l(x_N)
                            \end{bmatrix}
        \end{align}
    \hfill\qed
\end{definition}

\begin{definition}[Power Operator]
    \label{Def:PowerOperator}
    Consider the same $N\times 1$ block vector $x$.
    The power operator $P: \mb{R}^{MN} \rightarrow \mb{R}^{N}$ is defined as the mapping:
        \begin{align}
            \label{Equ:Def:PowerOperator}
            P[x]    \defeq  \col\{ \|x_1\|^2, \ldots, \|x_N\|^2\}
        \end{align}
    where $\|\cdot\|$ denotes the Euclidean norm of a vector.
    \hfill\qed
\end{definition}
We will use the power operator to study how error variances propagate after a specific
operator $T_A(\cdot)$ or $T_G(\cdot)$ is applied to a random vector.
We remark that we are using the notation ``$P[\cdot]$'' rather than ``$P(\cdot)$''
to highlight the fact that $P$ is a mapping from $\mb{R}^{MN}$ to a lower dimensional space $\mb{R}^{N}$.
In addition to the above three operators, we define the following aggregate vector of gradient
noise that depends on the state $x$:
    \begin{align}
        \label{Equ:PerformanceAnalysis:v_x}
        \bm{v}(x)   \defeq      \!-\col\Big\{\!
                                        \mu_1\! \sum_{l=1}^N c_{l1} \bm{v}_l(x_1),
                                        \ldots,
                                        \mu_N\! \sum_{l=1}^N c_{lN} \bm{v}_l(x_N)
                                        \!
                                \Big\}
    \end{align}
With these definitions, we can now represent the two combination steps
\eqref{Equ:PerformanceAnalysis:Diffusion_General_Noisy}
and \eqref{Equ:PerformanceAnalysis:Diffusion_General2_Noisy} as two
combination operators $T_{A_1}(\cdot)$ and $T_{A_2}(\cdot)$. We can also
represent the adaptation step
\eqref{Equ:PerformanceAnalysis:Diffusion_General1_Noisy}
by a gradient-descent operator perturbed by
the noise operator \eqref{Equ:PerformanceAnalysis:v_x}:
    \begin{align}
        \label{Equ:Def:T_G_hat}
        \widehat{\bm{T}}_G(x)     \defeq    T_G(x) + \bm{v}(x)
    \end{align}
We can view $\widehat{\bm{T}}_G(x)$ as a random operator that maps each input $x \in \mathbb{R}^{MN}$
into an $\mb{R}^{MN}$ random vector, and we use boldface letter to highlight this random
nature.
Let
    \begin{align}
        \bm{w}_i    \defeq  \col\{\bm{w}_{1,i}, \bm{w}_{2,i},\ldots,\bm{w}_{N,i}\}
    \end{align}
denote the vector that collects the estimators across all nodes.
Then, the overall diffusion adaptation steps
\eqref{Equ:PerformanceAnalysis:Diffusion_General_Noisy}--%
\eqref{Equ:PerformanceAnalysis:Diffusion_General2_Noisy} that update
$\bm{w}_{i-1}$ to $\bm{w}_i$ can be represented as a cascade composition of three operators:
    \begin{align}
        \label{Equ:Def:T_d_hat}
        \widehat{\bm{T}}_d(\cdot)     \defeq       T_{A_2} \circ \widehat{\bm{T}}_G \circ T_{A_1} (\cdot)
    \end{align}
where we use $\circ$ to denote the composition of any two operators, i.e.,
$T_1 \circ T_2 (x)   \defeq  T_1(T_2(x))$.
If there is no gradient noise, then the diffusion adaptation
operator \eqref{Equ:Def:T_d_hat} reduces to
    \begin{align}
        \label{Equ:Def:T_d}
        T_d(\cdot)     \defeq       T_{A_2} \circ T_G \circ T_{A_1} (\cdot)
    \end{align}
In other words, the diffusion adaptation over the entire network with and without gradient
noise can be described in the following compact forms:
    \begin{align}
        \label{Equ:Def:Diffusion_OperatorForm:Noisy}
        \bm{w}_i    &=   \widehat{\bm{T}}_d(\bm{w}_{i-1})
        \\
        \label{Equ:Def:Diffusion_OperatorForm:NoiseFree}
        w_i    &=   T_d(w_{i-1})
    \end{align}
Fig. \ref{Fig:Fig_DiffusionOperator:T_A_T_G} illustrates the role of the combination operator
$T_A(\cdot)$ (combination steps) and the gradient-descent operator $T_G(\cdot)$ (adaptation step).
The combination operator $T_A(\cdot)$ aggregates the estimates from the neighborhood (social learning),
while the gradient-descent operator $T_G(\cdot)$ incorporates information from the local gradient vector
(self-learning). In Fig. \ref{Fig:Fig_DiffusionOperator:OperatorDiffusion},
we show that each diffusion adaptation step can be represented as the cascade
composition of three operators,
with perturbation from the gradient noise operator.
    \begin{figure*}[t]
        \centerline{
            \subfigure[$T_{A_1}(\cdot)$, $T_{A_2}(\cdot)$ and $T_G(\cdot)$.]
            {\includegraphics[width=0.34\textwidth]{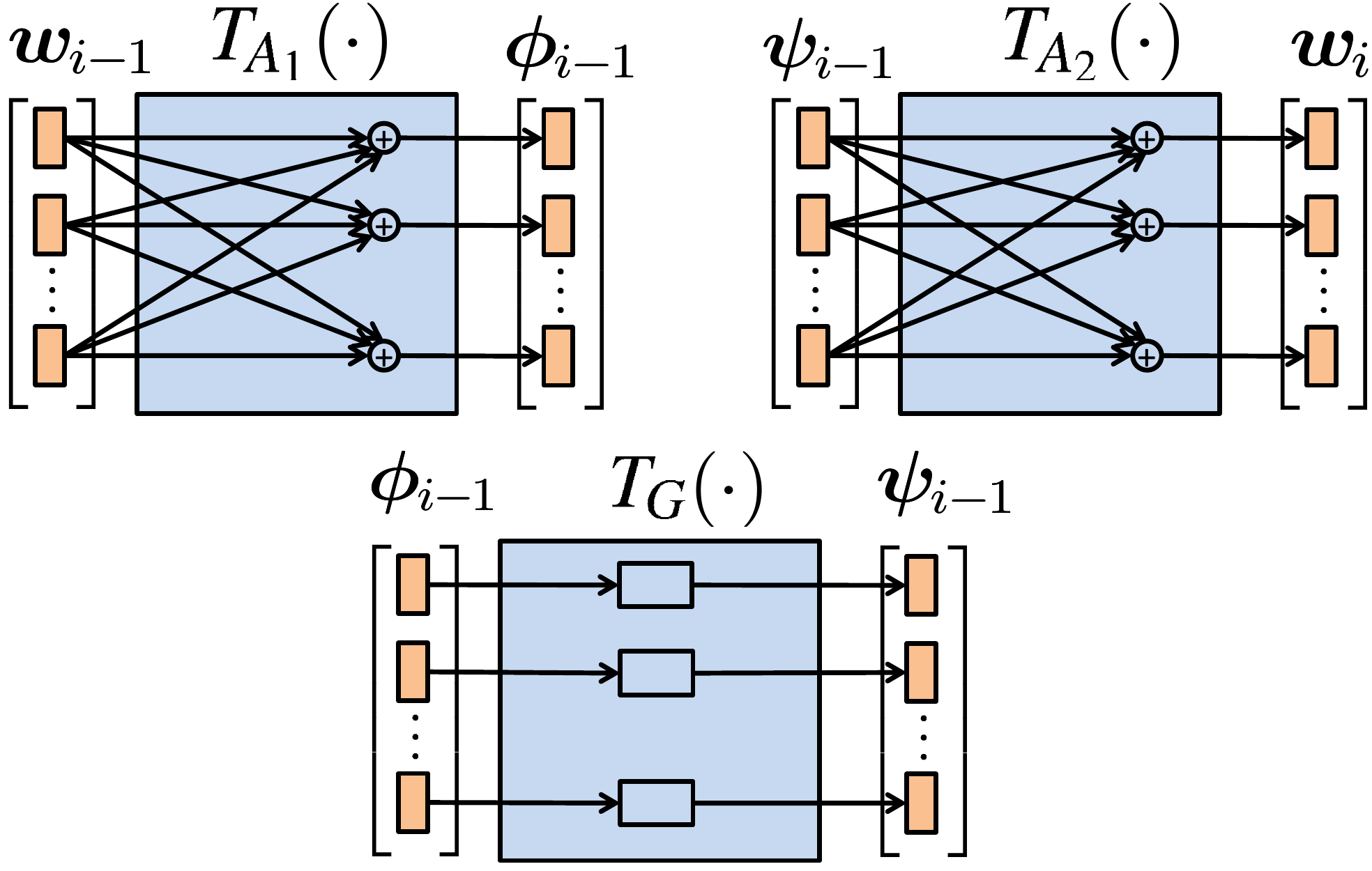}
            \label{Fig:Fig_DiffusionOperator:T_A_T_G}
            }
            \subfigure[Cascade representation of diffusion adaptation.]
            {\includegraphics[width=0.43\textwidth]{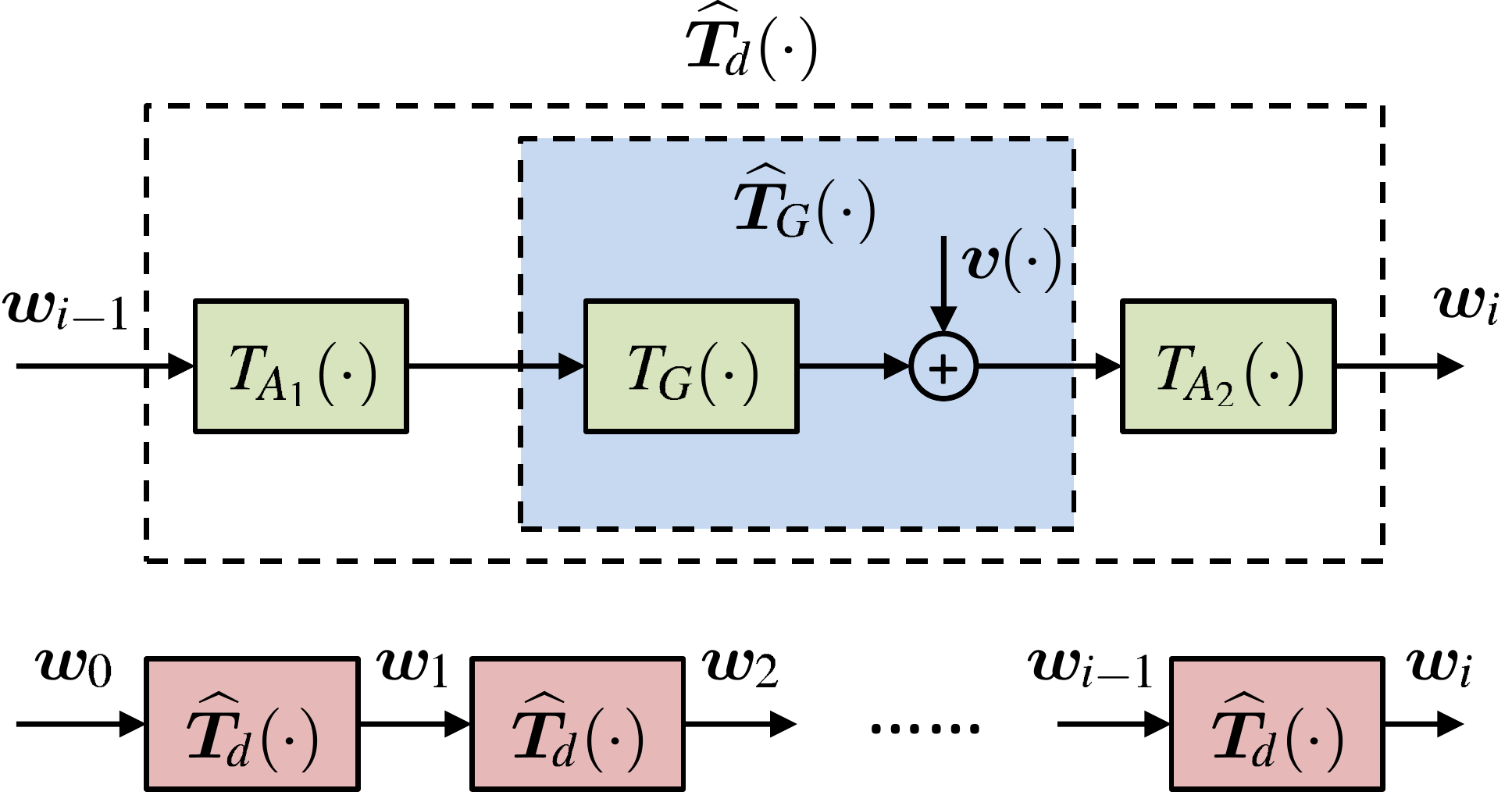}
            \label{Fig:Fig_DiffusionOperator:OperatorDiffusion}
            }
        }
        \caption{Representation of the diffusion adaptation strategy
        \eqref{Equ:PerformanceAnalysis:Diffusion_General_Noisy}--%
        \eqref{Equ:PerformanceAnalysis:Diffusion_General2_Noisy}
        in terms of operators.
        Each diffusion adaptation
        step can be viewed as a cascade composition of three operators: $T_{A_1}(\cdot)$,
        $T_{G}(\cdot)$, and $T_{A_2}(\cdot)$ with gradient perturbation $\bm{v}(\cdot)$.
        If $\bm{v}(\cdot)=0$, then
        $\widehat{\bm{T}}_d(\cdot)$
        becomes $T_d(\cdot)$.}
        \label{Fig:Fig_DiffusionOperator}
        \vspace{-\baselineskip}
    \end{figure*}
Next, in Lemma \ref{Lemma:PropertiesOperator}, we examine some of the properties of the operators $\{T_{A_1}, T_{A_2}, T_{G}\}$, which are proved in Appendix \ref{Appendix:Properties_Operators}.
    \begin{lemma}[Useful Properties]
        \label{Lemma:PropertiesOperator}
        Consider $N\times 1$ block vectors
        $x=\col\{x_1,\ldots,x_N \}$
        and $y=\col\{y_1,\ldots,y_N\}$
        with $M\times 1$ entries $\{x_k,y_k\}$. Then, the operators $T_A(\cdot)$, $T_G(\cdot)$
        and $P[\cdot]$ satisfy the following properties:
            \begin{enumerate}
                \item
                    (\underline{Linearity}): $T_A(\cdot)$ is a linear operator.
                \item
                    (\underline{Nonnegativity}): $P[x] \succeq 0$.
                \item
                    (\underline{Scaling}): For any scalar $a \in \mb{R}$, we have
                        \begin{align}
                            \label{Equ:Properties:PX_scalar}
                            P[ax]=a^2 P[x]
                        \end{align}
                \item
                    (\underline{Convexity}): suppose $x^{(1)},\ldots,x^{(K)}$ are $N\times 1$ block vectors formed in the same manner as $x$, and let $a_1,\ldots,a_K$ be non-negative real scalars that add up to one. Then,
                        \begin{align}
                            \label{Equ:Properties:PX_CvxComb}
                            P[a_1 x^{(1)}&+\cdots+ a_K x^{(K)}]
                            		\preceq a_1 P[x^{(1)}] + \cdots + a_K P[x^{(K)}]
                        \end{align}
                \item
                    (\underline{Additivity}): Suppose $\bm{x}=\col\{\bm{x}_1,\ldots,\bm{x}_N\}$ and
                    $\bm{y}=\col\{\bm{y}_1,\ldots,\bm{y}_N\}$ are $N\times 1$ block
                    random vectors that satisfy
                    $\E \bm{x}_k^T \bm{y}_k=0$ for $k=1,\ldots,N$. Then,
                        \begin{align}
                            \label{Equ:Properties:PX_Expt}
                            \E P[\bm{x}+\bm{y}] =   \E P[\bm{x}] + \E P[\bm{y}]
                        \end{align}
                \item
                    (\underline{Variance relations}):
                        \begin{align}
                            \label{Equ:Properties:PX_TA}
                            &P[T_A(x)]   \preceq     A^T P[x]
                            \\
                            \label{Equ:Properties:PX_TG}
                            &P[T_G(x)-T_G(y)]   \preceq  \Gamma^2 P[x-y]
                        \end{align}
					where
                        \begin{align}
                            \label{Equ:Properties:Gamma}
                            \Gamma      &\defeq     \diag\{\gamma_1,\ldots,\gamma_N\}       \\
                            \label{Equ:PerformanceAnalysis:gamma_k_def}
                            \gamma_k    &\defeq     \max
                                    				\{
                                    				    	|1-\mu_k \sigma_{k,\max}
                                    						|,\;
                                    						|1-\mu_k \sigma_{k,\min}
                                                            |
                                    				\}
                                    \\
                            \label{Equ:PerformanceAnalysis:Thm_MSS:sigma_min_def}
                            \sigma_{k,\min}    &\triangleq      \sum_{l=1}^N c_{lk} \lambda_{l,\min},
                                                                    \;
                            \sigma_{k,\max}     \triangleq      \sum_{l=1}^N c_{lk} \lambda_{l,\max}
                        \end{align}
                \item
                    (\underline{Block Maximum Norm}): The $\infty-$norm of $P[x]$ is the squared block maximum norm of $x$:
                        \begin{align}
                        	\label{Equ:Properties:PX_BMNorm}
                            \|P[x]\|_{\infty}   =       \|x\|_{b,\infty}^2
                                                \defeq  \big(\max_{1 \le k \le N}\|x_k\|\big)^2
                        \end{align}
                \item
                    (\underline{Preservation of Inequality}): Suppose
                    vectors $x$, $y$ and matrix $F$ have nonnegative entries, then
                    $x \preceq y$ implies $Fx \preceq  Fy$.\hfill\qed
            \end{enumerate}
    \end{lemma}

\subsection{Transient Analysis}
\label{Sec:TransientAnalysis}
Using the operator representation developed above,
we now analyze the transient behavior of the diffusion algorithm
\eqref{Equ:PerformanceAnalysis:Diffusion_General_Noisy}--\eqref{Equ:PerformanceAnalysis:Diffusion_General2_Noisy}.
From Fig. \ref{Fig:Fig_DiffusionOperator:OperatorDiffusion} and the previous discussion,
we know that the stochastic recursion $\bm{w}_i = \widehat{\bm{T}}_d(\bm{w}_{i-1})$
is a perturbed version of the noise-free recursion
$w_i = T_d(w_{i-1})$. Therefore, we first study the convergence of the noise free
recursion, and then analyze the effect of gradient perturbation on the stochastic recursion.

Intuitively, if recursion $w_{i}=T_d(w_{i-1})$ converges,
then it should converge to a vector $w_{\infty}$ that satisfies
	\begin{align}
		\label{Equ:PerformanceAnalysis:FixedPoint_NoiseFreeRecursion}
		w_{\infty}	=	T_d(w_{\infty})
	\end{align}
In other words, the vector $w_{\infty}$ should be a \emph{fixed point}
of the operator $T_d(\cdot)$\cite[p.299]{kreyszig1989introductory}.
We need to answer four questions pertaining to the fixed point.
First, does the fixed point exist? Second, is it unique?
Third, under which condition does the recursion
$w_{i}=T_d(w_{i-1})$ converge to the fixed point?
Fourth, how far is the fixed point $w_{\infty}$ away from the minimizer $w^o$ of \eqref{Equ:ProblemFormulation:J_glob_original}?
We answer the first two questions using the Banach Fixed Point Theorem (Contraction Theorem)
\cite[pp.2--9, pp.299--300]{kreyszig1989introductory}.
Afterwards, we study convergence under gradient perturbation.
The last question will be considered in the next subsection.
	\begin{definition}[Metric Space]
    	\label{Def:MetricSpace}
        A set $X$, whose elements we shall call points, is said to be a metric space if
        we can associate a real number $d(p,q)$ with  any two points $p$ and $q$ of $X$, such that
            \begin{enumerate}
                \item[(a)]
                    $d(p,q)>0$ if $p\neq q$; $d(p,p)=0$;
                \item[(b)]
                    $d(p,q)=d(q,p)$;
                \item[(c)]
                    $d(p,q)\le d(p,r) + d(r,q)$, for any $r \in X$.
            \end{enumerate}
        Any function $d(p,q)$ with these three properties is called a distance function, or
        a metric, and we denote a metric space $X$ with distance $d(\cdot,\cdot)$ as $(X,d)$.
        \hfill\qed
    \end{definition}
    \begin{definition}[Contraction]
        \label{Def:Contraction}
        Let $(X,d)$ be a metric space. A mapping $T: X \longrightarrow X$ is called a
        contraction on $X$ if there is a positive real number $\delta < 1$ such that
        $d(T(x),T(y))    \le \delta \cdot d(x,y)$ for all $x,y \in X$
    \end{definition}
    \begin{lemma}[Banach Fixed Point Theorem\cite{kreyszig1989introductory}]
        \label{Thm:BanachFixedPoint}
        Consider a metric space $(X,d)$, where $X \neq \emptyset$. Suppose that $X$ is complete%
        \footnote{A metric space $(X,d)$ is complete if any of its Cauchy sequences converges
        to a point in the space;
        a sequence $\{x_n\}$ is Cauchy in $(X,d)$ if $\forall \epsilon>0$,
        there exists $N$ such that
        $d(x_n,x_m)<\epsilon$ for all $n,\; m>N$.}
        and let $T:X\rightarrow X$ be a contraction. Then, $T$ has precisely one fixed point.
        \hfill\qed
    \end{lemma}
As long as we can prove that the diffusion operator $T_d(\cdot)$ is a contraction, i.e.,
for any two points $x,y \in \mb{R}^{MN}$, after we apply the operator $T_d(\cdot)$,
the distance between $T_d(x)$ and $T_d(y)$
scales down by a scalar that is uniformly bounded away from one, then
the fixed point $w_{\infty}$ defined in \eqref{Equ:PerformanceAnalysis:FixedPoint_NoiseFreeRecursion}
exists and is unique. We now proceed to show that $T_d(\cdot)$ is a contraction operator in
$X=\mb{R}^{MN}$ when the step-size parameters $\{\mu_k\}$ satisfy certain conditions.
    \begin{theorem}[Fixed Point]
    		\label{Thm:ExistenceUniqueFixedPoint}
    		Suppose the step-size parameters $\{\mu_k\}$ satisfy the following conditions 
    			\begin{align}
				\label{Equ:PerformanceAnalysis:Contraction_StepSize}
					0	<	\mu_k	<	\frac{2}{\sigma_{k,\max}},
		        \qquad
		        k=1,2,\ldots,N
			\end{align}
		Then, there exists a unique fixed point $w_{\infty}$ for the unperturbed diffusion
		operator $T_d(\cdot)$ in \eqref{Equ:Def:T_d}.
    \end{theorem}
    \begin{proof}
{
    		Let $x=\col\{x_1,\ldots,x_N\} \in \mb{R}^{MN \times 1}$ be
			formed by stacking $M \times 1$ vectors $x_1,\ldots,x_N$
			on top of each other. Similarly, let $y=\col\{y_1,\ldots,y_N\}$.
			The distance function $d(x,y)$ that we will use is induced from the block maximum
			norm \eqref{Equ:Properties:PX_BMNorm}:
					$d(x,y)	=	\|x-y\|_{b,\infty}	=	\max_{1 \le k \le N} \|x_k-y_k\|$.
			From the definition of the diffusion operator $T_d(\cdot)$ in \eqref{Equ:Def:T_d},
			we have
			{\allowdisplaybreaks
				\begin{align}
					P[T_d(x)-T_d(y)]	
									&\overset{(a)}{=}		
											P\big[
												T_{A_2}
												\big(
													T_G \circ T_{A_1}(x)
													-
													T_G \circ T_{A_1}(y)
												\big)
											\big]
											\nonumber\\
									&\overset{(b)}{\preceq}
											A_2^T
											P\big[
												T_G \circ T_{A_1}(x)
												-
												T_G \circ T_{A_1}(y)
											\big]
											\nonumber\\
					\label{Equ:PerformanceAnalysis:Contraction_intermediate1}
									&\overset{(c)}{\preceq}	
											A_2^T \Gamma^2
											P[T_{A_1}(x) - T_{A_1}(y)]
											\nonumber\\
									&\overset{(d)}{=}		
											A_2^T \Gamma^2 P[T_{A_1}(x-y)]
											\nonumber\\
									&\overset{(e)}{\preceq}	
											A_2^T \Gamma^2 A_1^T P[x-y]
				\end{align}
			}%
			where steps (a) and (d) are because of the linearity of $T_{A_1}(\cdot)$ and $T_{A_2}(\cdot)$,
			steps (b) and (e) are because of the variance relation property \eqref{Equ:Properties:PX_TA},
			and step (c) is due to the variance relation property \eqref{Equ:Properties:PX_TG}.
			Taking the $\infty-$norm of both sides of \eqref{Equ:PerformanceAnalysis:Contraction_intermediate1},
			we have
				\begin{align}
					\label{Equ:PerformanceAnalysis:Contraction_intermediate2}
					\|P[T_d(x)-T_d(y)]\|_{\infty}	&\le	\|A_2^T \Gamma^2 A_1^T\|_{\infty}
												\cdot \|P[x-y]\|_{\infty}
												\nonumber\\
											&\le	\|\Gamma\|_{\infty}^2
												\cdot
												\|P[x-y]\|_{\infty}
				\end{align}
			where, in the second inequality, we used the fact that $\|A_1^T\|_{\infty}=\|A_2^T\|_{\infty}=1$
			since $A_1^T$ and $A_2^T$ are right-stochastic matrices.
			Using property \eqref{Equ:Properties:PX_BMNorm},
			we can conclude from \eqref{Equ:PerformanceAnalysis:Contraction_intermediate2} that:
					$\|T_d(x)-T_d(y)\|_{b,\infty}	\le	\|\Gamma\|_{\infty} \cdot \|x-y\|_{b,\infty}$.
			Therefore, the operator $T_d(\cdot)$ is a contraction if $\|\Gamma\|_{\infty}<1$,
			which, by substituting \eqref{Equ:Properties:Gamma}--\eqref{Equ:PerformanceAnalysis:gamma_k_def},
			becomes
				\begin{align}
					|1-\mu_k\sigma_{k,\max}|<1, \quad |1-\mu_k\sigma_{k,\min}|<1, 
					\quad k=1,\ldots,N
					\nonumber
				\end{align}
			and we arrive at the condition \eqref{Equ:PerformanceAnalysis:Contraction_StepSize} on the step-sizes
			In other words, if condition \eqref{Equ:PerformanceAnalysis:Contraction_StepSize} holds
			for each $k=1,\ldots,N$, then $T_d(\cdot)$ is a contraction operator.
			By Lemma \ref{Thm:BanachFixedPoint}, the operator $T_d(\cdot)$
			will have a unique fixed point $w_{\infty}$ that satisfies equation
			\eqref{Equ:PerformanceAnalysis:FixedPoint_NoiseFreeRecursion}.
}
    \end{proof}

Given the existence and uniqueness of the fixed point, the third question to answer
is if recursion $w_{i}=T_d(w_{i-1})$
converges to this fixed point. The answer is affimative under
\eqref{Equ:PerformanceAnalysis:Contraction_StepSize}.
However, we are not going to study this question separately. Instead, we will analyze
the convergence of the more demanding noisy recursion \eqref{Equ:Def:Diffusion_OperatorForm:Noisy}.
Therefore, we now study how fast and how close the successive estimators $\{\bm{w}_i\}$ generated
by recursion \eqref{Equ:Def:Diffusion_OperatorForm:Noisy} approach $w_{\infty}$.
Once this issue is addressed, we will then examine how close $w_{\infty}$ is to the desired $w^o$. 
Introduce the following mean-square-perturbation (MSP) vector at time $i$:
	\begin{align}
		\label{Equ:PerformanceAnalysis:MSP_i}
		\mathrm{MSP}_i	&\defeq	\E P[\bm{w}_i - w_{\infty}]
	\end{align}
The $k$-th entry of $\mathrm{MSP}_i$ characterizes how far away
the estimate $\bm{w}_{k,i}$ at node $k$ and time $i$ is from $w_{k,\infty}$ in the mean-square
sense. To study the closeness of $\bm{w}_i$ to $w_\infty$, we shall study how the
quantity $\mathrm{MSP}_i$ evolves over time.
By \eqref{Equ:Def:Diffusion_OperatorForm:Noisy},
\eqref{Equ:PerformanceAnalysis:FixedPoint_NoiseFreeRecursion}
and the definitions of $\widehat{\bm{T}}_d(\cdot)$ and $T_d(\cdot)$ in
\eqref{Equ:Def:T_d_hat} and \eqref{Equ:Def:T_d}, we obtain
{\allowdisplaybreaks
	\begin{align}
		\mathrm{MSP}_i			
						&=		\E P[\bm{w}_i - w_{\infty}]	
								\nonumber\\
						&\quad=		\E
								P\big[
									T_{A_2} \circ \widehat{\bm{T}}_G \circ T_{A_1}(\bm{w}_{i-1})
									-
									T_{A_2} \circ T_G \circ T_{A_1}(w_{\infty})
								\big]
								\nonumber\\
						&\quad\overset{(a)}{=}
								\E
								P\big[
									T_{A_2}
									\big(
										\widehat{\bm{T}}_G \circ T_{A_1}(\bm{w}_{i-1})
										-
										T_G \circ T_{A_1}(w_{\infty})
									\big)
								\big]
								\nonumber\\
						&\quad\overset{(b)}{\preceq}
								A_2^T
								\E
								P\big[
										\widehat{\bm{T}}_G \circ T_{A_1}(\bm{w}_{i-1})
										-
										T_G \circ T_{A_1}(w_{\infty})
								\big]
								\nonumber\\
						&\quad\overset{(c)}{=}
								A_2^T
								\E
								P\big[
										T_G\big(T_{A_1}(\bm{w}_{i-1})\big)
										-
										T_G\big(T_{A_1}(w_{\infty})\big)
										+\bm{v}\big(T_{A_1}(\bm{w}_{i-1})\big)
								\big]
								\nonumber\\
						&\quad\overset{(d)}{=}
								A_2^T
								\big\{
									\E
									P\big[
										T_G\big(T_{A_1}(\bm{w}_{i-1})\big)
										-
										T_G\big(T_{A_1}(w_{\infty})\big)
									\big]
									+
									\E
									P\big[
										\bm{v}\big(T_{A_1}(\bm{w}_{i-1})\big)
									\big]
								\big\}
								\nonumber\\
						&\quad\overset{(e)}{\preceq}
								A_2^T
								\Gamma^2
								\E
								P[
									T_{A_1}(\bm{w}_{i-1})
									-
									T_{A_1}(w_{\infty})
								]
								+
								A_2^T
								\E
                                P[
									\bm{v}(T_{A_1}(\bm{w}_{i-1}))
								]
								\nonumber\\
						&\quad\overset{(f)}{\preceq}
								A_2^T \Gamma^2 A_1^T \cdot
								\E P[\bm{w}_{i-1}-w_{\infty}]
								+
								A_2^T \E P[\bm{v}(T_{A_1}(\bm{w}_{i-1}))]
								\nonumber\\
		\label{Equ:PerformanceAnalysis:MSP_intermediate1}
						&\quad=
								A_2^T \Gamma^2 A_1^T \cdot
								\mathrm{MSP}_{i-1}
								+
								A_2^T \E P[\bm{v}(T_{A_1}(\bm{w}_{i-1}))]
	\end{align}
}%
where step (a) is by the linearity of $T_{A_1}(\cdot)$, steps (b) and (f) are by property
\eqref{Equ:Properties:PX_TA}, step (c) is by the substitution of
\eqref{Equ:Def:T_G_hat}, step (d) is by Property 5 in Lemma \ref{Lemma:PropertiesOperator}
and assumption \eqref{Assumption:GradientNoise:ZeroMean_Uncorrelated},
and step (e) is by \eqref{Equ:Properties:PX_TG}.
To proceed with the analysis, we establish the following lemma to bound
the second term in \eqref{Equ:PerformanceAnalysis:MSP_intermediate1}.
	\begin{lemma}[Bound on Gradient Perturbation]
		\label{Lemma:Bound_on_GradientPerturbation}
		It holds that
			\begin{align}
				\label{Equ:Lemma:Bound_on_GradientPerturbation}
				\E P[\bm{v}(T_{A_1}(\bm{w}_{i-1}))]
					&\preceq	4\alpha \lambda_{\max}^2 \|C\|_1^2
								\!\cdot\!
								\Omega^2
								A_1^T
								\!\cdot\!
								\E P[\bm{w}_{i\!-\!1}\!-\!w_{\infty}]
								\!+\!\|C\|_1^2 \Omega^2 b_v
			\end{align}
		where 
			\begin{align}
				\label{Equ:Lemma:BGP:lambda_max_def}
				\lambda_{\max}		\;\defeq\;&		\max_{1 \le k \le N} \lambda_{k,\max}	\\
				b_v				\;\defeq\;&		4\alpha \lambda_{\max}^2 A_1^T
											P[w_{\infty}-\mathds{1}_N \otimes w^o]\
											\nonumber\\
				\label{Equ:Lemma:BGP:xi_v_def}
											&+
											\max_{1 \le k \le N}
											\{
												2\alpha \|\nabla_w J_k(w^o)\|^2
												+
												\sigma_v^2
											\}
											\\
                \label{Equ:Lemma:BGP:Omega_def}
				\Omega			\;\defeq\;&		\diag
											\{
												\mu_1, \ldots, \mu_N
											\}
			\end{align}
	\end{lemma}
	\begin{proof}
{
		By the definition of $\bm{v}(\bm{x})$ in \eqref{Equ:PerformanceAnalysis:v_x}
		with $\bm{x}=T_{A_1}(\bm{w}_{i-1})$ being a random vector, we get
			\begin{align}
				\label{Equ:Appendix:P_vx_intermediate1}
				\E P[\bm{v}(\bm{x})]		&=		
													\begin{bmatrix}
		            									\mu_1^2
		            									\E\big\|
		            										\sum_{l=1}^N c_{l1} \bm{v}_l(\bm{x}_1),				
		            									\big\|^2\\
		            									\vdots\\
		            									\mu_N^2
		            									\E\big\|
		            										\sum_{l=1}^N c_{lN} \bm{v}_l(\bm{x}_N)
		            									\big\|^2
													\end{bmatrix}
			\end{align}
		For each block in \eqref{Equ:Appendix:P_vx_intermediate1}, using Jensen's inequality,
		we have
		{\allowdisplaybreaks
			\begin{align}
				\E\big\|
					\sum_{l=1}^N &c_{lk} \bm{v}_l(\bm{x}_k)
				\big\|^2
									\nonumber\\
						&=			\big(\sum_{l=1}^N c_{lk}\big)^2
									\cdot
									\E\big\|
										\sum_{l=1}^N
										\frac{c_{lk}}{\sum_{l=1}^N c_{lk}} \bm{v}_l(\bm{x}_k)
									\big\|^2
									\nonumber\\
						&\le			\big(\sum_{l=1}^N c_{lk}\big)^2
									\cdot
									\sum_{l=1}^N
										\frac{c_{lk}}{\sum_{l=1}^N c_{lk}}
									\E\|
										\bm{v}_l(\bm{x}_k)
									\|^2
									\nonumber\\
				\label{Equ:Appendix:E_norm_sum_v_x_bound_intermediate1}
						&\le			\|C\|_1
									\sum_{l=1}^N
									c_{lk}
									\big[
										\alpha \E\|\nabla_w J_l(\bm{x}_k)\|^2 + \sigma_v^2
									\big]
			\end{align}
		}%
		where $\|\cdot\|_1$ denotes the maximum absolute column sum, and in the last step,
		we used \eqref{Assumption:GradientNoise:Norm}. Using
		\eqref{Equ:Appendix:MeanValueTheorem},
			\begin{align}
				\nabla_w J_l&(\bm{x}_k)	
		                        =	\nabla_w J_l(w^o)
									+
									\big[
										\int_{0}^{1}
										\nabla_w^2
										J_l\big(
											w^o + t(\bm{x}_k-w^o)
										\big)
										dt
									\big]
									(\bm{x}_k - w^o)
			\end{align}
		From \eqref{Equ:ConvergenceAnalysis:H_lki_Bounds} and the norm
		inequality $\|x+y\|^2 \le 2\|x\|^2 + 2\|y\|^2$,
		we obtain
			\begin{align}
				\|\nabla_w J_l(\bm{x}_k)\|^2	
								&\le	2\|\nabla_w J_l(w^o)\|^2
									\!+\!
									2\lambda_{l,\max}^2
									\!\cdot\!
									\|\bm{x}_k - w^o\|^2
									\nonumber\\
				\label{Equ:Appendix:Bound_gradient_Jl_xk}
								&\le	2\|\nabla_w J_l(w^o)\|^2
									\!+\!
									2\lambda_{\max}^2
									\!\cdot\!
									\|\bm{x}_k - w^o\|^2
			\end{align}
		Substituting \eqref{Equ:Appendix:Bound_gradient_Jl_xk} into
		\eqref{Equ:Appendix:E_norm_sum_v_x_bound_intermediate1}, we obtain
			\begin{align}
				\E\big\|
					\sum_{l=1}^N c_{lk} \bm{v}_l(\bm{x}_k)
				\big\|^2		
							&\le	\|C\|_1\!
									\sum_{l=1}^N\!
									c_{lk}
									\big[
										2\alpha \lambda_{\max}^2
										\E\|\bm{x}_k \!-\! w^o\|^2
										\!+\!
										2\alpha \|\nabla_w J_l(w^o)\|^2 \!+\! \sigma_v^2
									\big]
									\nonumber\\
				\label{Equ:Appendix:E_norm_sum_v_x_bound_intermediate2}
							&\le		2\alpha \lambda_{\max}^2 \|C\|_1^2								
									\cdot\E\|\bm{x}_k - w^o\|^2
									+
									\|C\|_1^2
									\cdot
									\overline{\sigma}_v^2
			\end{align}
		where $\overline{\sigma}_v^2
					\defeq	\displaystyle\max_{1 \le l \le N}
							\{2\alpha \|\nabla_w J_l(w^o)\|^2 + \sigma_v^2\}$.
		Substituting \eqref{Equ:Appendix:E_norm_sum_v_x_bound_intermediate2}
		and $\bm{x}=T_{A_1}(\bm{w}_{i-1})$ into
		\eqref{Equ:Appendix:P_vx_intermediate1} leads to
		{\allowdisplaybreaks
			\begin{align}
				\E P[\bm{v}(T_{A_1}(\bm{w}_{i-1}))]	
								&\preceq		\Omega^2
											\big\{
												2\alpha\|C\|_1^2
												\lambda_{\max}^2
												\cdot
												\E
												P[T_{A_1}(\bm{w}_{i-1})-\mathds{1}_N \otimes w^o]+
												\|C\|_1^2
												\overline{\sigma}_v^2 \mathds{1}_N
											\big\}
											\nonumber\\
								&\overset{(a)}{=}
											\Omega^2
											\big\{
												2\alpha\|C\|_1^2
												\lambda_{\max}^2
												\cdot
												\E
												P\big[
													T_{A_1}(\bm{w}_{i-1})
													-
													T_{A_1}(\mathds{1}_N \otimes w^o)
												\big]
												+
												\|C\|_1^2
												\overline{\sigma}_v^2 \mathds{1}_N
											\big\}
											\nonumber\\
								&\overset{(b)}{=}
											\Omega^2
											\big\{
												2\alpha\|C\|_1^2
												\lambda_{\max}^2
												\cdot
												\E
												P\big[
													T_{A_1}
													\big(
														\bm{w}_{i-1}
														-
														\mathds{1}_N \otimes w^o
													\big)
												\big]
												+
												\|C\|_1^2
												\overline{\sigma}_v^2 \mathds{1}_N
											\big\}
											\nonumber\\
								&\overset{(c)}{\preceq}
											\Omega^2
											\big\{
												2\alpha\|C\|_1^2
												\lambda_{\max}^2
												A_1^T
												\cdot
												\E
												P[
													\bm{w}_{i-1}
													-
													\mathds{1}_N \otimes w^o
												]
												+
												\|C\|_1^2
												\overline{\sigma}_v^2 \mathds{1}_N
											\big\}
											\nonumber\\
								&\overset{(d)}{=}			
											\Omega^2
											\Big\{\!
												2\alpha\|C\|_1^2
												\lambda_{\max}^2
												A_1^T
												\!\cdot\!
												4
												\E
												P\Big[\!
													\frac{\bm{w}_{i-1} 
													\!-\! w_{\infty}}{2}
													\!+\!
													\frac{
														w_{\infty}
														\!-\!
														\mathds{1}_N \otimes w^o
														}{2}
														\!
												\Big]
												+
												\|C\|_1^2
												\overline{\sigma}_v^2 \mathds{1}_N
											\Big\}
											\nonumber\\
								&\overset{(e)}{\preceq}		
											\Omega^2
											\big\{
												2\alpha\|C\|_1^2
												\lambda_{\max}^2
												A_1^T
												\cdot
												\big(
												2
												\E
												P[
													\bm{w}_{i-1} \!-\! w_{\infty}
												]
												\!+\!
												2P[
													w_{\infty}
													\!-\!
													\mathds{1}_N \otimes w^o
												]
												\big)
												\!+\!
												\|C\|_1^2
												\overline{\sigma}_v^2 \mathds{1}_N
											\big\}
											\nonumber\\
								&=			
												4\alpha\|C\|_1^2
												\lambda_{\max}^2
												\!\cdot\!
												\Omega^2
												A_1^T
												\!\cdot\!
												\E
												P[
													\bm{w}_{i-1} \!-\! w_{\infty}
												]
												\!+\!
												\|C\|_1^2\Omega^2
												\!\cdot\!
												b_v
			\end{align}
		}%
		where step (a) is due to the fact that $A_1^T$ is right-stochastic so that
		$T_{A_1}(\mathds{1}_N \otimes w^o)=\mathds{1}_N \otimes w^o$,
		step (b) is because of the linearity of $T_{A_1}(\cdot)$, step (c)
		is due to property \eqref{Equ:Properties:PX_TA},
		step (d) is a consequence of Property 3 of Lemma \ref{Lemma:PropertiesOperator},
		and step (e)
		is due to the convexity property \eqref{Equ:Properties:PX_CvxComb}.
	}
	\end{proof}
\noindent
Substituting \eqref{Equ:Lemma:Bound_on_GradientPerturbation} into
\eqref{Equ:PerformanceAnalysis:MSP_intermediate1},
we obtain
	\begin{align}
		\label{Equ:PerformanceAnalysis:MSP_InequalityRecursion}
		\boxed{
		\mathrm{MSP}_i		\preceq		A_2^T \Gamma_d A_1^T
										\cdot
										\mathrm{MSP}_{i-1}
										+
										\|C\|_1^2 \cdot A_2^T \Omega^2 b_v
		}
	\end{align}
where
	\begin{align}
		\label{Equ:PerformanceAnalysis:Gamma_d_def}
		\Gamma_d		&\defeq	\Gamma^2 + 4 \alpha \lambda_{\max}^2 \|C\|_1^2 \cdot \Omega^2	
	\end{align}
The following theorem gives the stability conditions on the inequality recursion
\eqref{Equ:PerformanceAnalysis:MSP_InequalityRecursion} and derives
both asymptotic and non-asymptotic bounds for MSP.
	\begin{theorem}[Mean-Square Stability and Bounds]
		\label{Thm:MSStability}
		Suppose $A_2^T \Gamma_d A_1^T$ is a stable matrix, i.e.,
		$\rho(A_2^T \Gamma_d A_1^T)<1$. Then, the following
		non-asymptotic bound holds for all $i \ge 0$:
			\begin{align}
				\label{Equ:Thm_MSS:NonAsymptoticBound}
				\mathrm{MSP}_i	\preceq	(A_2^T \Gamma_d A_1^T)^i
										[
											\mathrm{MSP}_0
											-
											\mathrm{MSP}_{\infty}^{\mathrm{ub}}
										]
										+
										\mathrm{MSP}_{\infty}^{\mathrm{ub}}
			\end{align}
		where $\mathrm{MSP}_{\infty}^{\mathrm{ub}}$ is the asymptotic upper bound
		on MSP defined as
			\begin{align}
				\label{Equ:Thm_MSS:MSP_infty_ub_def}
				\mathrm{MSP}_{\infty}^{\mathrm{ub}}
							\defeq		\|C\|_1^2
										(I_N - A_2^T \Gamma_d A_1^T)^{-1}
										A_2^T \Omega^2 b_v
			\end{align}
		And, as $i \rightarrow \infty$, we have the following asymptotic bound
			\begin{align}
				\label{Equ:Thm_MSS:AsymptoticBound}
				\limsup_{i \rightarrow  \infty}	\mathrm{MSP}_i	
								\preceq	\mathrm{MSP}_{\infty}^{\mathrm{ub}}
			\end{align}
		Furthermore, a sufficient condition that guarantees the stability of
		the matrix $A_2^T \Gamma_d A_1^T$ is that
			\begin{align}
				\label{Equ:Thm_MSS:StepSize_Condition}
				0	\!<\!	\mu_k	\!<\!
										\min\Big\{
										\frac{\sigma_{k,\max}}
											{
												\sigma_{k,\max}^2
												\!\!+\!
												4\alpha \lambda_{\max}^2 \|C\|_1^2
											},\;
										\frac{\sigma_{k,\min}}
											{
												\sigma_{k,\min}^2
												\!\!+\!
												4\alpha \lambda_{\max}^2 \|C\|_1^2
											}
									\Big\}
			\end{align}
		for all $k=1, \ldots, N$, where $\sigma_{k,\max}$ and $\sigma_{k,\min}$
		were defined earlier in \eqref{Equ:PerformanceAnalysis:Thm_MSS:sigma_min_def}.
	\end{theorem}
	\begin{proof}
{
		Iterating inequality \eqref{Equ:PerformanceAnalysis:MSP_InequalityRecursion},
		we obtain
			\begin{align}
				\label{Equ:Appendix:MSP_i_Bound_intermediate}
				\mathrm{MSP}_i	\preceq	(A_2^T \Gamma_d A_1^T)^{i}
										\mathrm{MSP}_0
										+
										\|C\|_1^2
										\cdot
										\big[
											\sum_{j=0}^{i-1}
											(A_2^T \Gamma_d A_1^T)^{j}
										\big]
										A_2^T \Omega^2 b_v								
			\end{align}
		For the second term in \eqref{Equ:Appendix:MSP_i_Bound_intermediate},
		we note that $(I + X + \cdots + X^{i-1}) (I-X) = I - X^i$.
		If $X$ is a stable matrix so that $(I-X)$ is invertible, then 
		it leads to
				$\sum_{j=0}^{i-1} X^j	=	(I-X^i)(I-X)^{-1}$.
		Using this relation and given that the matrix $A_2^T \Gamma_d A_1^T$ is stable,
		we can express \eqref{Equ:Appendix:MSP_i_Bound_intermediate} as
			\begin{align}
				\mathrm{MSP}_i	&\preceq	(A_2^T \Gamma_d A_1^T)^{i}
										\mathrm{MSP}_0
										+
										\|C\|_1^2
										\cdot
										\big[
											I_N
											\!-\!
											(A_2^T \Gamma_d A_1^T)^{i}
										\big]
										(I_N \!-\!  A_2^T \Gamma_d A_1^T)^{-1}
										A_2^T \Omega^2 b_v			
										\nonumber\\
				\label{Equ:Appendix:MSP_i_Bound_final}
								&=		(A_2^T \Gamma_d A_1^T)^i
										[
											\mathrm{MSP}_0
											-
											\mathrm{MSP}_{\infty}^{\mathrm{ub}}
										]
										+
										\mathrm{MSP}_{\infty}^{\mathrm{ub}}
			\end{align}
		Letting $i \rightarrow \infty$ on both sides of the above inequality, we get
				$\displaystyle\limsup_{i \rightarrow \infty} \mathrm{MSP}_i
								\preceq	\mathrm{MSP}_{\infty}^{\mathrm{ub}}$.
		In the last step, we need to show the conditions on the step-sizes $\{\mu_k\}$
		that guarantee stability of the matrix $A_2^T \Gamma_d A_1^T$.
		Note that the spectral radius of a matrix is upper bounded by its matrix norms.
		Therefore,
			\begin{align}
				\rho(A_2^T \Gamma_d A_1^T)	
										&\le	\|A_2^T \Gamma_d A_1^T\|_{\infty}
												\nonumber\\		
										&\le	\|A_2^T\|_{\infty}
												\cdot
												\|\Gamma_d\|_{\infty}
												\cdot
												\|A_1^T\|_{\infty}		
												\nonumber\\
										&=	\|\Gamma_d\|_{\infty}
											\nonumber\\
										&=	\big\|
												\Gamma^2
												+
												4\alpha\lambda_{\max}^2\|C\|_1^2
												\cdot
												\Omega^2
											\big\|_{\infty}
											\nonumber
			\end{align}
		If the right-hand side of the above inequality is strictly less than one, then
		the matrix $A_2^T\Gamma_d A_2^T$ is stable. Using
		\eqref{Equ:Properties:Gamma}--\eqref{Equ:PerformanceAnalysis:gamma_k_def}, this condition
		is satisfied by the following
		quadratic inequalities on $\mu_k$ :
			\begin{align}
					(1-\mu_k\sigma_{k,\max})^2 + \mu_k^2 \cdot 4\alpha \lambda_{\max}^2 \|C\|_1^2 < 1
					\\
					(1-\mu_k\sigma_{k,\min})^2 + \mu_k^2 \cdot 4\alpha \lambda_{\max}^2 \|C\|_1^2 < 1
			\end{align}
		for all $k=1,\ldots,N$.
		Solving the above inequalities, we obtain condition \eqref{Equ:Thm_MSS:StepSize_Condition}.
}
	\end{proof}
\noindent
The non-asymptotic bound \eqref{Equ:Thm_MSS:NonAsymptoticBound} characterizes how
the MSP at each node evolves over time. It shows that the MSP
converges to steady state at a geometric rate determined by the spectral radius of
the matrix $A_2^T\Gamma_d A_1^T$. The transient term is determined by
the difference between the initial MSP and the steay-state MSP.
At steady state, the MSP is upper bounded by $\mathrm{MSP}_{\infty}^{\mathrm{ub}}$.
We now examine closely how small the steady-state MSP can be for small step-size
parameters $\{\mu_k\}$. Taking the $\infty-$norm of both sides of
\eqref{Equ:Thm_MSS:AsymptoticBound} and using the relation
        $(I_N - A_2^T \Gamma_d A_1^T)^{-1}   =   \sum_{j=0}^{\infty} (A_2^T \Gamma_d A_1^T)^j$,
we obtain
    \begin{align}
        \|\mathrm{MSP}_{\infty}^{\mathrm{ub}}\|_{\infty}
                &=      \big\|
                            \|C\|_1^2
                            \cdot
							(I_N - A_2^T \Gamma_d A_1^T)^{-1}
                            \cdot
							A_2^T \Omega^2 b_v
                        \big\|_{\infty}
                        \nonumber\\
                &\le    \|C\|_1^2
                            \cdot
                        \Big(
                            \sum_{j=0}^{\infty}
                            \|A_2^T\|_{\infty}^j
                                \cdot
                            \|\Gamma_d\|_{\infty}^j
                                \cdot
                            \|A_1^T\|_{\infty}^j
                        \Big)
                        \cdot
                        \|A_2^T\|_{\infty}
                        \cdot
                        \|\Omega\|_{\infty}^2
                        \cdot
                        \|b_v\|_{\infty}
                        \nonumber\\
                &\overset{(a)}{\le}
                        \|C\|_1^2
                            \cdot
                        \Big(
                            \sum_{j=0}^{\infty}
                            \|\Gamma_d\|_{\infty}^j
                        \Big)
                        \cdot
                        \big(\max_{1 \le k \le N} \mu_k\big)^2
                        \cdot
                        \|b_v\|_{\infty}
                        \nonumber\\
        \label{Equ:PerformanceAnalysis:MSP_inf_ub_SmallStepSizes_intermediate1}
                &=      \frac{\|C\|_1^2 \cdot \|b_v\|_{\infty}}{1-\|\Gamma_d\|_{\infty}}
                        \cdot
                        \big(\max_{1 \le k \le N} \mu_k\big)^2
    \end{align}
where step (a) is because $A_1^T$ and $A_2^T$ are right-stochastic matrices so that their $\infty-$norms
(maximum absolute row sum) are one.
Let $\mu_{\max}$ and $\mu_{\min}$ denote the maximum and minimum values of $\{\mu_k\}$, respectively,
and let $\beta \defeq \mu_{\min}/\mu_{\max}$.
For sufficiently small step-sizes, by the definitions of $\Gamma_d$ and $\Gamma$
in \eqref{Equ:PerformanceAnalysis:Gamma_d_def} and \eqref{Equ:Properties:Gamma}, we
have
    \begin{align}
        \|\Gamma_d\|_{\infty}   \;\le\;&    \|\Gamma\|_{\infty}^2
                                            +
                                            4\alpha\lambda_{\max}\|C\|_1^2
                                            \cdot
                                            \|\Omega\|_{\infty}^2\
                                            \nonumber\\
                                \;\overset{(a)}{=}\;&
                                            \max_{1 \le k \le N}
                                                    \{
                                                        |1-\mu_k\sigma_{k,\min}|^2
                                                    \}
                                            +
                                            4\alpha\lambda_{\max}\mu_{\max}^2\|C\|_1^2
                                            \nonumber\\
        \label{Equ:PerformanceAnalysis:Gamma_d_bound}
                                \;\le\;&    1\!-\!2\mu_{\min}\sigma_{\min} \!+\!
                                            \mu_{\max}^2
                                            (\sigma_{\max}^2\!+\!4\alpha\lambda_{\max}\|C\|_1^2)
                                            \nonumber\\
                                \;=\;&      1\!-\!2\beta\mu_{\max}\sigma_{\min} \!+\!
                                            \mu_{\max}^2
                                            (\sigma_{\max}^2\!+\!4\alpha\lambda_{\max}\|C\|_1^2)
    \end{align}
where $\sigma_{\max}$ and $\sigma_{\min}$ are the maximum and minimum values of $\{\sigma_{k,\max}\}$
and $\{\sigma_{k,\min}\}$, respectively, and step (a) holds for sufficiently small step-sizes.
Note that \eqref{Equ:PerformanceAnalysis:MSP_inf_ub_SmallStepSizes_intermediate1} is a monotonically
increasing function of $\|\Gamma_d\|_{\infty}$. Substituting
\eqref{Equ:PerformanceAnalysis:Gamma_d_bound} into
\eqref{Equ:PerformanceAnalysis:MSP_inf_ub_SmallStepSizes_intermediate1},
we get
    \begin{equation}
        \label{Equ:PerformanceAnalysis:MSP_AsymptoticBound_SmallStepSizes}
        \boxed{
        \begin{split}
        \limsup_{i \rightarrow\infty} \|\mathrm{MSP}_{i}\|_{\infty}
                        &\le
                                    \|\mathrm{MSP}_{\infty}^{\mathrm{ub}}\|_{\infty}
                        \le        \frac{\|C\|_1^2 \cdot \|b_v\|_{\infty} \cdot \mu_{\max}}
                                        {
                                            2\beta\sigma_{\min} \!-\! \mu_{\max}
                                            (\sigma_{\max}^2\!+\!4\alpha\lambda_{\max}\|C\|_1^2)
                                        }
                        \!\sim\!        O(\mu_{\max})
        \end{split}
        }
    \end{equation}
{\color{black}
Note that, for sufficiently small step-sizes, the right-hand side of
\eqref{Equ:PerformanceAnalysis:MSP_AsymptoticBound_SmallStepSizes} is approximately
$\frac{\|C\|_1^2 \cdot \|b_v\|_{\infty}}{2\beta\sigma_{\min}} \mu_{\max}$, which
is on the order of $O(\mu_{\max})$.
In other words, the steady-state MSP can be made
be arbitrarily small for small step-sizes, and
the estimators $\bm{w}_{i}=\col\{\bm{w}_{1,i},\ldots,\bm{w}_{N,i}\}$
will be close to the fixed point $w_{\infty}$ (in the mean-square sense) even under gradient
perturbations. 
To understand how close the estimate $\bm{w}_{k,i}$ at each node $k$ is
to the Pareto-optimal solution $w^o$, a natural question to consider is how close the fixed point
$w_{\infty}$ is to $\mathds{1}_N \otimes w^o$, which we study next.

\subsection{Bias Analysis}
\label{Sec:BiasAnalysis}

Our objective is to examine how large $\|\mathds{1}_N \otimes w^o-w_{\infty}\|^2$ is
when the step-sizes are small. We carry out the analysis in two steps:
first, we derive an expression for $\tilde{w}_{\infty} \defeq  \mathds{1}_N \otimes w^o-w_{\infty}$,
and then we derive the conditions that guarantee small bias.

To begin with, recall that $w_{\infty}$ is the fixed point of $T_d(\cdot)$, to which
the recursion $w_{i}=T_d(w_{i-1})$ converges.
Also note that $T_d(\cdot)$ is an operator representation of the recursions
\eqref{Equ:PerformanceAnalysis:Diffusion_General}--\eqref{Equ:PerformanceAnalysis:Diffusion_General2}.
We let $i\rightarrow \infty$ on both sides of
\eqref{Equ:PerformanceAnalysis:Diffusion_General}--\eqref{Equ:PerformanceAnalysis:Diffusion_General2}
and obtain
{\setlength{\jot}{-5pt}
    \begin{align}
            \label{Equ:PerformanceAnalysis:FixedPointEquation1_Comb1}
			{\phi}_{k,\infty}	&=	 \displaystyle\sum_{l = 1 }^N a_{1,lk} \; {w}_{l,\infty}
                                        \\
            \label{Equ:PerformanceAnalysis:FixedPointEquation1_Adapt}
			{\psi}_{k,\infty}	&=									
								\displaystyle{\phi}_{k,\infty}
								-
								\mu_k \sum_{l =1}^N
								c_{lk}
									{\nabla}_w J_l({\phi}_{k,\infty})
								\\
            \label{Equ:PerformanceAnalysis:FixedPointEquation1_Comb2}
			{w}_{k,\infty}		&=	 \displaystyle\sum_{l = 1}^N a_{2,lk}\;
                                            {\psi}_{l,\infty}
	\end{align}
}%
where $w_{k,\infty}$, $\phi_{k,\infty}$ and $\psi_{k,\infty}$ denote the limits
of $w_{k,i}$, $\phi_{k,i}$ and $\psi_{k,i}$ as $i \rightarrow \infty$, respectively.
Introduce the following bias vectors at node $k$
    \begin{align}
        \label{Equ:PerformanceAnalysis:w_psi_phi_tilde_inf_def}
        \tilde{w}_{k,\infty}    \defeq  w^o\!-\!w_{k,\infty},
        \;
        \tilde{\phi}_{k,\infty}    \defeq  w^o\!-\!\phi_{k,\infty},
        \;
        \tilde{\psi}_{k,\infty}    \defeq  w^o\!-\!\psi_{k,\infty}
    \end{align}
Subtracting each equation of
\eqref{Equ:PerformanceAnalysis:FixedPointEquation1_Comb1}--\eqref{Equ:PerformanceAnalysis:FixedPointEquation1_Comb2}
from $w^o$ and using relation $\nabla_w J_l(\phi_{k,\infty})   =   \nabla_w J_l(w^o) - H_{lk,\infty} \tilde{\phi}_{k,\infty}$ 
that can be derived from 
Lemma \ref{Lemma:MeanValueTheorem} in Appendix \ref{Appendix:Properties_Operators},
we obtain
    \begin{align}
            \label{Equ:PerformanceAnalysis:FixedPointEquation_Comb1}
			\tilde{\phi}_{k,\infty}	&=	 \displaystyle\sum_{l = 1 }^N a_{1,lk} \; \tilde{w}_{l,\infty}
                                        \\
            \label{Equ:PerformanceAnalysis:FixedPointEquation_Adapt2}
	        \tilde{\psi}_{k,\infty} &=   \Big[
	                                        I_M \!-\! \mu_k \sum_{l=1}^N c_{lk} H_{lk,\infty}
	                                    \Big]
	                                    \tilde{\phi}_{k,\infty}
	                                    \!+\!
	                                    \mu_k
	                                    \sum_{l=1}^N
	                                    c_{lk}
	                                    \nabla_w J_l(w^o)
								\\
            \label{Equ:PerformanceAnalysis:FixedPointEquation_Comb2}
			\tilde{w}_{k,\infty}		&=	 \displaystyle\sum_{l = 1}^N a_{2,lk}\;
                                            \tilde{\psi}_{l,\infty}
	\end{align}                      
where $H_{lk,\infty}$ is a positive semi-definite symmetric  matrix
defined as
        \begin{align}
                \label{Equ:ConvergenceAnalysis:H_kinf}
        		{H}_{lk,\infty}	
        				\;\triangleq\;&          \int_{0}^{1}
                    							\nabla_w^2
                    							J_l
                                                \big(
                                                    w^o
                                                    \!-\!
                                                    t
                                                    \sum_{l=1}^{N} a_{1,lk}
                                                    \tilde{w}_{l,\infty})
                                                \big) dt
        \end{align}
Introduce the following global vectors and matrices
    \begin{align}
        \tilde{w}_{\infty}  &\defeq      \mathds{1}_N \otimes w^o-w_{\infty}
                            =           \col\{
                                                    \tilde{w}_{1,\infty},
                                                    \ldots,
                                                    \tilde{w}_{N,\infty}
                                            \}
                                            \\
        \label{Equ:PerformanceAnalysis:A_cal_1_2_def}
        \mathcal{A}_1   &\defeq     A_1 \otimes I_M,
        \quad
        \mathcal{A}_2   \defeq     A_2 \otimes I_M,
        \quad
        \mathcal{C}     \defeq      C \otimes I_M,
        \quad
        \\
        \label{Equ:PerformanceAnalysis:M_cal_def}
        \mathcal{M}     &\defeq      \mathrm{diag}\{\mu_1,\ldots,\mu_N\}\otimes I_M
        \\
        \label{Equ:ConvergenceAnalysis:R_inf}
        \mathcal{R}_{\infty}	
                        &\defeq
                                \sum_{l=1}^N
        						\mathrm{diag}
        						\big\{
        							 c_{l1} {H}_{l1,\infty},
        							\cdots,
        							c_{lN} {H}_{lN,\infty}
        						\big\},
        	\quad
        \\
        \label{Equ:PerformanceAnalysis:g_o_def}
        g^o             &\defeq
                                \col
                                \{
                                    \nabla_w J_1(w^o), \ldots, \nabla_w J_N(w^o)
                                \}
    \end{align}
Then, expressions \eqref{Equ:PerformanceAnalysis:FixedPointEquation_Comb1},
\eqref{Equ:PerformanceAnalysis:FixedPointEquation_Comb2} and
\eqref{Equ:PerformanceAnalysis:FixedPointEquation_Adapt2} 
lead to
    \begin{align}
    \label{Equ:PerformanceAnalysis:Thm_Convergence:w_tilde_infty_Expr}
    \boxed{
    \tilde{w}_{\infty}
                                        =
                                            \big[
                                                    I_{MN}
                                                    \!-\!
                                                    \mathcal{A}_2^T
                                                    \left(I_{MN}
                                                    \!-\!
                                                    \mathcal{M}
                                                     \mathcal{R}_{\infty}\right)
                                                    \mathcal{A}_1^T
                                            \big]^{-1}
                                            \!\!\!
                                            \mathcal{A}_2^T \mathcal{M} \mathcal{C}^T g^o
    }
    \end{align}

    \begin{theorem}[Bias at Small Step-sizes]
    	\label{Thm:Bias_Small_StepSize}
    	Suppose that $A_2^TA_1^T$ is a regular right-stochastic matrix, so
	   that its eigenvalue of largest magnitude is one with multiplicity one,
	   and all other eigenvalues are strictly smaller than one. Let $\theta^T$ denote
	   the left eigenvector of $A_2^TA_1^T$ of eigenvalue one. Furthermore, assume
        the following condition holds:
    		\begin{align}
    			\label{Equ:ConvergenceAnalysis:Thm:ConsensusCondition}
    			&\theta^T A_2^T \Omega C^T = c_0 \mathds{1}^T
    		\end{align}
    	where
        $\Omega \triangleq \mathrm{diag}\{\mu_1,\ldots,\mu_N\}$ was defined
        earlier in Lemma \ref{Lemma:Bound_on_GradientPerturbation}, and  $c_0$ is some constant.
    	Then,
            \begin{align}
                \label{Equ:ConvergenceAnalysis:Thm:SmallBias_mu}
                \|\tilde{w}_{\infty}\|^2
                =
                \|
                    \mathds{1}_N \otimes w^o - w_{\infty}
                \|^2   \sim    O(\mu_{\max}^2)
            \end{align}
    \end{theorem}
    \begin{proof}
        See Appendix \ref{Appendix:BiasSmallStepSizes}.
    \end{proof}
\noindent
Therefore, as long as the network is connected (not necessarily fully connected)
and condition \eqref{Equ:ConvergenceAnalysis:Thm:ConsensusCondition} holds,
the bias would become arbitrarily small.
For condition \eqref{Equ:ConvergenceAnalysis:Thm:ConsensusCondition} to hold, one
choice is to require the matrices $A_1^T$ and $A_2^T$ to be doubly stochastic, and
all nodes to use the same step-size $\mu$, namely, $\Omega=\mu I_N$. In that case,
the matrix $A_1^TA_2^T$ is doubly-stochastic so that the left eigenvector of eigenvalue
one is $\theta^T=\mathds{1}^T$ and \eqref{Equ:ConvergenceAnalysis:Thm:ConsensusCondition} holds.

Finally, we combine the results from
Theorems \ref{Thm:MSStability} and \ref{Thm:Bias_Small_StepSize}
to bound the mean-square-error (MSE) of the estimators $\{\bm{w}_{k,i}\}$
from the desired Pareto-optimal solution $w^o$.
Introduce the $N \times 1$ MSE vector
    \begin{align}
        \mathrm{MSE}_i  &\defeq    \E P[\tilde{\bm{w}}_i]
        							\nonumber\\
                        &=          \E P[\mathds{1}_N \otimes w^o - \bm{w}_i]        
                        				\nonumber\\
        \label{Equ:PerformanceAnalysis:MSE_def}
                        &=          \col
                                    \big\{
                                        \E\|\tilde{\bm{w}}_{1,i}\|^2, \ldots,\E\|\tilde{\bm{w}}_{N,i}\|^2
                                    \big\}
    \end{align}
Using Properties 3--4 in Lemma \ref{Lemma:PropertiesOperator},
we obtain
    \begin{align}
        \mathrm{MSE}_i  &=          \E
                                    P\big[
                                        2
                                        \big(
                                            \frac{\mathds{1}_N \otimes w^o \!-\! w_{\infty}}{2}
                                            \!+\!
                                            \frac{w_{\infty} \!-\! \bm{w}_i}{2}
                                        \big)
                                    \big]
                                    \nonumber\\
                        &\preceq    2P[\tilde{w}_{\infty}] \!+\! 2\;\E P[w_{\infty} \!-\! \bm{w}_i]
                                   \nonumber\\
        \label{Equ:PerformanceAnalysis:MSE_vector_def}
                        &=          2P[\tilde{w}_{\infty}] \!+\! 2 \; \mathrm{MSP}_i
    \end{align}
Taking the $\infty-$norm of both sides of above inequality and using
property \eqref{Equ:Properties:PX_BMNorm}, we obtain
    \begin{align}
        \limsup_{i \rightarrow \infty} \|\mathrm{MSE}_i\|_{\infty}
            &\le
                2\|P[\tilde{w}_{\infty}]\|_{\infty}
                +
                2 \limsup_{i\rightarrow \infty} \|\mathrm{MSP}_i\|_{\infty}
                \nonumber\\
            &=
                2\|\tilde{w}_{\infty}\|_{b,\infty}^2
                +
                2 \limsup_{i\rightarrow \infty} \|\mathrm{MSP}_i\|_{\infty}
                \nonumber\\
        \label{Equ:PerformanceAnalysis:MSE_AsymptoticBound_SmallStepsize}
            &\sim
                O(\mu_{\max}^2) + O(\mu_{\max})
    \end{align}
where in the last step, we used \eqref{Equ:PerformanceAnalysis:MSP_AsymptoticBound_SmallStepSizes}
and \eqref{Equ:ConvergenceAnalysis:Thm:SmallBias_mu}, and the fact that
all vector norms are equivalent. Therefore, as the step-sizes become small,
the MSEs become small and the estimates $\{\bm{w}_{k,i}\}$ get
arbitrarily close to the Pareto-optimal solution $w^o$. We also observe
that, for small step-sizes,
the dominating steady-state error is MSP, which is caused by the gradient noise
and is on the order of $O(\mu_{\max})$. On the other hand, the bias term
is a high order component, i.e., $O(\mu_{\max}^2)$, and can be ignored.

{\color{black}
The fact that the bias term $\tilde{w}_{\infty}$ is small also gives us a useful
approximation for $\mc{R}_{\infty}$ in \eqref{Equ:ConvergenceAnalysis:R_inf}.
Since $\tilde{w}_{\infty}=\col\{\tilde{w}_{1,\infty},\ldots,\tilde{w}_{N,\infty}\}$ is small
for small step-sizes, the matrix $H_{lk,\infty}$ defined in \eqref{Equ:ConvergenceAnalysis:H_kinf}
can be approximated as $H_{lk,\infty} \approx \nabla_w^2 J_l(w^o)$.
Then, by definition \eqref{Equ:ConvergenceAnalysis:R_inf}, we have
    \begin{equation}
        \label{Equ:PerformanceAnalysis:R_inf_Approx_SmallBias}
        \boxed{
                \mathcal{R}_{\infty}
                                \approx    \sum_{l=1}^N
                                            \mathrm{diag}
                                            \big\{
                                                    c_{l1} \nabla_w^2 J_l(w^o), \ldots, c_{lN} \nabla_w^2
                                                    J_l(w^o)
                                        \big\}
        }
    \end{equation}
Expressing \eqref{Equ:PerformanceAnalysis:R_inf_Approx_SmallBias} is useful for
evaluating closed-form expressions of the steady-state MSE in sequel.
}

\subsection{Steady-State Performance}
\label{Sec:Steady-State Performance}
So far, we derived inequalities \eqref{Equ:PerformanceAnalysis:MSE_AsymptoticBound_SmallStepsize}
to bound the steady-state performance,
and showed that, for small step-sizes,
the solution at each node $k$ approaches the same Pareto-optimal point $w^o$.
In this section, we derive closed-form expressions (rather than bounds) for the steady-state MSE
at small step-sizes. Introduce the error vectors%
\footnote{
\label{Footnote:Notation_Error_Quantities}
{\color{black}
In this paper, we always use the notation $\tilde{w}=w^o-w$ to denote
the error relative to $w^o$. For the error between $w$ and the fixed point $w_{\infty}$,
we do not define a separate notation, but instead write $w_{\infty}-w$ explicitly
to avoid confusion.}}%
	\begin{align}		
		\label{Equ:ConvergenceAnalysis:Def_errorQuantities}
		\tilde{\bm{\phi}}_{k,i}	\triangleq	w^o\!-\!\bm{\phi}_{k,i},
		\;
		\tilde{\bm{\psi}}_{k,i}	\triangleq	w^o\!-\!\bm{\psi}_{k,i},
		\;
		\tilde{\bm{w}}_{k,i}	\triangleq	w^o\!-\!\bm{w}_{k,i}
	\end{align}
and the following global random quantities
    {\setlength{\jot}{-1pt}
	\begin{align}
		\label{Equ:ConvergenceAnalysis:global_error_vector}
		\tilde{\bm w}_{i}		\triangleq&
							\col\{
								\tilde{\bm w}_{1,i},
								\ldots,
								\tilde{\bm w}_{N,i}
							\}
							\\
		\label{Equ:ConvergenceAnalysis:D_i_minus_1}
		\bm{\mathcal{R}}_{i-1}	
					\defeq&	\sum_{l=1}^N
						\mathrm{diag}
						\big\{
							 c_{l1} \bm{H}_{l1,i-1},
							\cdots,
							c_{lN} \bm{H}_{lN,i-1}
						\big\}
						\\
		\label{Equ:ConvergenceAnalysis:H_kim1}
		\bm{H}_{lk,i-1}	
				\triangleq&          \int_{0}^{1}
            							\!\!\nabla_w^2
            							J_l
                                        \Big(
                                            w^o \!-\! t\sum_{l=1}^{N} a_{1,lk} \tilde{\bm{w}}_{l,i\!-\!1}
                                        \Big) dt
                        \\
		\label{Equ:ConvergenceAnalysis:G_i}
		\bm{g}_i
					\defeq&	\sum_{l=1}^N
						\mathrm{col}
						\big\{
							 c_{l1}\bm{v}_{l}
                            ({\bm{\phi}}_{1,i\!-\!1}),
							\cdots,
							c_{lN}\bm{v}_{l}
                            ({\bm{\phi}}_{N,i\!-\!1})
						\big\}
	\end{align}
	}%
Then, extending the derivation from \cite[Sec. IV A]{chen2011TSPdiffopt}, we can establish that
    \begin{align}
		\label{Equ:ConvergenceAnalysis:ErrorRecursion_final}
			\tilde{\bm w}_i	=	\mathcal{A}_2^T
							[I_{MN} \!-\! \mathcal{M}\bm{\mathcal{R}}_{i-1}]
							\mathcal{A}_1^T
							\tilde{\bm w}_{i-1}
							\!+\!
							\mathcal{A}_2^T \mathcal{M} \mathcal{C}^T g^o
							\!+\!
							\mathcal{A}_2^T \mathcal{M} \bm{g}_i
	\end{align}
According to \eqref{Equ:PerformanceAnalysis:MSE_AsymptoticBound_SmallStepsize},
the error $\tilde{\bm{w}}_{k,i}$ at each node $k$ would be small for small
step-sizes and after long enough time. In other words, $\bm{w}_{k,i}$ is close
to $w^o$. And recalling from
\eqref{Equ:PerformanceAnalysis:Diffusion_General_Noisy} that $\bm{\phi}_{k,i-1}$
is a convex combination of $\{\bm{w}_{l,i}\}$, we conclude that
the quantities $\{\bm{\phi}_{l,i-1}\}$
are also close to $w^o$.
Therefore,
we can approximate $\bm{H}_{lk,i-1}$, $\bm{\mathcal{R}}_{i-1}$ and $\bm{g}_i$
in \eqref{Equ:ConvergenceAnalysis:D_i_minus_1}--\eqref{Equ:ConvergenceAnalysis:G_i}
by
    \begin{align}
            \label{Equ:PerformanceAnalysis:H_lki1_Approx_SmallStepSize}
            \bm{H}_{lk,i\!-\!1} &\!\approx\!
					                       \int_{0}^{1}
					                       \!\!
                                            \nabla_w^2
                                            J_l(w^o)
                                            dt
                                \!=\!           \nabla_w^2
                                            J_l(w^o)
                                            \\
            \label{Equ:PerformanceAnalysis:R_i1_Approx_SmallStepSize}
            \bm{\mathcal{R}}_{i\!-\!1}
                            &\!\approx\!	\sum_{l=1}^N
                            			\!\!
                                            \mathrm{diag}
                                            \big\{
                                                    c_{l1} \nabla_w^2 J_l(w^o), \ldots, c_{lN} \nabla_w^2
                                                    J_l(w^o)
                                        \big\}
                                        \!\approx\!	\mc{R}_{\infty}
    \end{align}
Then, the error recursion \eqref{Equ:ConvergenceAnalysis:ErrorRecursion_final}
can be approximated by
    \begin{align}
		\label{Equ:ConvergenceAnalysis:Approx_ErrorRecursion_final}
        \boxed{
			\tilde{\bm w}_i	=	\mathcal{A}_2^T
							[I_{MN} \!-\! \mathcal{M}{\mathcal{R}}_{\infty}]
							\mathcal{A}_1^T
							\tilde{\bm w}_{i-1}
							\!+\!
							\mathcal{A}_2^T \mathcal{M} \mathcal{C}^T g^o
							\!+\!
							\mathcal{A}_2^T \mathcal{M} \bm{g}_i
        }
	\end{align}

First, let us examine the behavior of $\E\tilde{\bm{w}}_{i}$.
Taking expectation of both sides of recursion
\eqref{Equ:ConvergenceAnalysis:Approx_ErrorRecursion_final}, we obtain
	\begin{align}
		\label{Equ:ConvergenceAnalysis:MeanErrorRelation_Approx_SmallStepSize}
			\mathbb{E}\tilde{\bm w}_i	
						=	\mathcal{A}_2^T
							[I_{MN}-\mathcal{M}{\mathcal{R}}_{\infty}]
							\mathcal{A}_1^T
							\E\tilde{\bm w}_{i-1}
							+
							\mathcal{A}_2^T \mathcal{M} \mathcal{C}^T g^o
	\end{align}
This recursion converges when the matrix
$\mathcal{A}_2^T[I_{MN}-\mathcal{M}{\mathcal{R}}_{\infty}]\mathcal{A}_1^T$ is stable,
which is guaranteed by \eqref{Equ:PerformanceAnalysis:Contraction_StepSize}
(see Appendix C of \cite{chen2011TSPdiffopt}). Let $i \rightarrow \infty$ on both
sides of \eqref{Equ:ConvergenceAnalysis:MeanErrorRelation_Approx_SmallStepSize} so that
    \begin{equation}
        \label{Equ:PerformanceAnalysis:Ew_inf}
        \boxed{
        	\begin{split}
        \mathbb{E}\tilde{\bm w}_{\infty}    &\defeq \lim_{i\rightarrow\infty}
                                                    \mathbb{E}\tilde{\bm w}_i
                                                    \\
                                            &=       \big[
                                                    I_{MN}
                                                    \!-\!
                                                    \mathcal{A}_2^T
                                                    \left(I_{MN}
                                                    \!-\!
                                                    \mathcal{M}
                                                     \mathcal{R}_{\infty}\right)
                                                    \mathcal{A}_1^T
                                                    \big]^{-1}
                                                    \!\!\!
                                                    \mathcal{A}_2^T \mathcal{M} \mathcal{C}^T g^o
            \end{split}
        }
    \end{equation}
Note that $\mathbb{E}\tilde{\bm w}_{\infty}$ coincides with
\eqref{Equ:PerformanceAnalysis:Thm_Convergence:w_tilde_infty_Expr}.
{\color{black}
By Theorem \ref{Thm:Bias_Small_StepSize}, we know that the squared norm
of this expression
is on the order of $O(\mu_{\max}^2)$ at small step-sizes --- see
\eqref{Equ:ConvergenceAnalysis:Thm:SmallBias_mu}.
}
Next, we derive closed-form expressions for the MSEs, i.e., $\E\|\tilde{\bm{w}}_{k,i}\|^2$.
Let $R_v$ denote the covariance matrix of $\bm{g}_i$ evaluated at $w^o$:
			\begin{align}
				\label{Equ:ConvergenceAnalysis:GradientNoise_PreciseModel}
				R_v         
                            \;=\;&          \mb{E}
                                            \Big\{\!
                                            \Big[\!
                                                \sum_{l=1}^N
						                          \mathrm{col}
						                          \big\{
							                             c_{l1}\bm{v}_{l,i}(w^o),
							                             \cdots,
							                             c_{lN}\bm{v}_{l,i}(w^o)
						                          \big\}\!
                                            \Big]
                                            \Big[\!
                                                \sum_{l=1}^N
						                          \mathrm{col}
						                          \big\{
							                             c_{l1}\bm{v}_{l,i}(w^o),
							                             \cdots,
							                             c_{lN}\bm{v}_{l,i}(w^o)
						                          \big\}\!
                                            \Big]^T\!
                                            \Big\}
			\end{align}
In practice, we can evaluate $R_v$ from the expressions
of $\{\bm{v}_{l,i}(w^o)\}$.
Equating the squared \emph{weighted} Euclidean ``norm'' of both sides of
\eqref{Equ:ConvergenceAnalysis:Approx_ErrorRecursion_final}, applying the expectation operator with
assumption \eqref{Assumption:GradientNoise:ZeroMean_Uncorrelated},
and following the same line of reasoning
from \cite{chen2011TSPdiffopt}, we can establish the following approximate variance relation
at small step-sizes:
    \begin{align}
		\label{Equ:PerformanceAnalysis:WeightedEnergyConservation_Relation_SmallStepSizes}
			\mathbb{E}\|\tilde{\bm{w}}_{i}\|_\Sigma^2
						\approx&\;
							\mathbb{E}\|\tilde{\bm{w}}_{i-1}\|_{{\Sigma}'}^2
							+
							\mathrm{Tr}
        					(
        						\Sigma	
        						\mathcal{A}_2^T \mathcal{M}
        						R_v
        						\mathcal{M} \mathcal{A}_2				
        					)
        					+
                            \mathrm{Tr}
                            \{
                                \Sigma
                                \mathcal{A}_2^T \mathcal{M} \mathcal{C}^T g^o
                                (\mathcal{A}_2^T \mathcal{M} \mathcal{C}^T g^o)^T
                            \}
                            \nonumber\\
                            &+
							2
								(
								\mathcal{A}_2^T
								\mathcal{M}
								\mathcal{C}^T
								g^o
								)^T
								\Sigma
								\mathcal{A}_2^T
								\left(I_{MN}\!-\!\mathcal{M}{\mathcal{R}}_{\infty}\right)
								\mathcal{A}_1^T 
								\mathbb{E}\tilde{\bm w}_{i-1}
                            \\
        \label{Equ:PerformanceAnalysis:Sigma_PrimePrime_SmallStepSizes}
			{\Sigma}'	\approx&\;
                            \mathcal{A}_1
							\left(I_{MN}\!-\!\mathcal{M}{\mathcal{R}}_{\infty}\right)
							\mathcal{A}_2
							\Sigma
							\mathcal{A}_2^T
							\left(I_{MN}\!-\!\mathcal{M}{\mathcal{R}}_{\infty}\right)
							\mathcal{A}_1^T
	\end{align}
where $\Sigma$ is a positive semi-definite weighting matrix that we are free to choose.
Let $\sigma=\mathrm{vec}(\Sigma)$ denote the vectorization operation that
stacks the columns of a matrix $\Sigma$ on top of each other. We shall use
the notation $\|x\|_{\sigma}^2$ and $\|x\|_{\Sigma}^2$
interchangeably.
Following the argument from \cite{chen2011TSPdiffopt},
we can rewrite \eqref{Equ:PerformanceAnalysis:WeightedEnergyConservation_Relation_SmallStepSizes} as
    \begin{align}
        \label{Equ:PerformanceAnalysis:EnergyConservation_Final}
        \mathbb{E}\|\tilde{\bm{w}}_{i}\|_{\sigma}^2
            &\approx     \mathbb{E}\|\tilde{\bm{w}}_{i-1}\|_{F\sigma}^2
                        +
                        r^T\sigma
                        +
                        \sigma^T
                        Q \;
                        \mathbb{E}\tilde{\bm{w}}_{i-1}
   	\end{align}
where
   	\begin{align}
				F	&\triangleq
										\mathcal{A}_1
										[I_{MN}\!-\!\mathcal{M}{\mathcal{R}}_{\infty}]
										\mathcal{A}_2
									\otimes
										\mathcal{A}_1
										[I_{MN}\!-\!\mathcal{M}{\mathcal{R}}_{\infty}]
										\mathcal{A}_2
									\\
                r   &\triangleq         \mathrm{vec}\!
        								\left(
        									\mathcal{A}_2^T \mathcal{M}
        									R_v
        									\mathcal{M} \mathcal{A}_2	
        								\right)
                                        \!+\!
                                        \mathcal{A}_2^T \mathcal{M} \mathcal{C}^T g^o
                                        \!\otimes\!
                                        \mathcal{A}_2^T \mathcal{M} \mathcal{C}^T \!g^o
                                        \\
                Q   &\triangleq         2
                                            \mathcal{A}_2^T
                                    		(I_{MN}\!-\!\mathcal{M}{\mathcal{R}}_{\infty})
                                    		\mathcal{A}_1^T
                                        \otimes
                                    		\mathcal{A}_2^T
                                    		\mathcal{M}
                                    		\mathcal{C}^T
                                    		g^o
    \end{align}
We already established that
$\E \tilde{\bm{w}}_{i-1}$ on the right-hand side of
\eqref{Equ:PerformanceAnalysis:EnergyConservation_Final}
converges to its limit $\E\tilde{\bm{w}}_{\infty}$ under condition
\eqref{Equ:PerformanceAnalysis:Contraction_StepSize}.
And, it was shown in \cite[pp.344-346]{Sayed08} that
such recursion converges to a steady-state value if the matrix $F$ is stable,
i.e., $\rho(F)<1$. This condition is guaranteed when the step-sizes are sufficiently small
(or chosen according to \eqref{Equ:PerformanceAnalysis:Contraction_StepSize})
--- see the proof in Appendix C of \cite{chen2011TSPdiffopt}.
Letting $i \rightarrow \infty$ on both sides of expression
\eqref{Equ:PerformanceAnalysis:EnergyConservation_Final}, we obtain:
    \begin{align}
        \label{Equ:PerformanceAnalysis:EnergyConservation_inf}
        \boxed{
        \lim_{i \rightarrow \infty} \mathbb{E}\|\tilde{\bm{w}}_{i}\|_{(I-F)\sigma}^2
            \approx     \left(r+Q\;\mathbb{E}\tilde{\bm{w}}_{\infty}\right)^T\sigma
        }
    \end{align}
We can now resort to
\eqref{Equ:PerformanceAnalysis:EnergyConservation_inf} and use it to evaluate various
performance
metrics by choosing proper weighting matrices $\Sigma$ (or $\sigma$).
For example,
the MSE of any node $k$ can be obtained by computing
$\lim_{i \rightarrow \infty} \mathbb{E}\|\tilde{\bm{w}}_{i}\|_{T}^2$
with a block weighting matrix $T$ that has an identity matrix at block $(k,k)$ and zeros elsewhere:
		$\displaystyle\lim_{i \rightarrow \infty}\mathbb{E}\|\tilde{\bm{w}}_{k,i}\|^2	
                =	\lim_{i \rightarrow \infty}\mathbb{E}\|\tilde{\bm{w}}_{i}\|_{T}^2$.
Denote the vectorized version of this matrix by
		$t_k	\triangleq	\mathrm{vec}(\mathrm{diag}(e_k)\otimes I_M)$,
where $e_k$ is a vector whose $k$th entry is one and zeros elsewhere.
Then, if we select $\sigma$ in \eqref{Equ:PerformanceAnalysis:EnergyConservation_inf}
as $\sigma = (I-F)^{-1}t_k$,
the term on the left-hand side becomes the desired
$\lim_{i \rightarrow \infty}\mathbb{E}\|\tilde{\bm{w}}_{k,i}\|^2$
and the MSE for node $k$  is therefore given by:
	\begin{align}
		\label{Equ:ConvergenceAnalysis:MSE_k}
        \boxed{
			\mathrm{MSE}_k	\defeq
                                    \lim_{i \rightarrow \infty} \E\|\tilde{\bm{w}}_{k,i}\|^2
                            \approx	
								\left(r+Q\;\mathbb{E}\tilde{\bm{w}}_{\infty}\right)^T
								(I\!-\!F)^{-1} t_k
        }
	\end{align}
If we are interested in the average network MSE,
then it is given by
    \begin{align}
    		\label{Equ:PerformanceAnalysis:MSE_Network}
        \overline{\mathrm{MSE}}
                    \triangleq	\frac{1}{N} \sum_{k=1}^N \mathrm{MSE}_k
    \end{align}

\section{Application to Collaborative Decision Making}
\label{Sec:Simulation}
We illustrate one application of the framework developed in the previous sections to the problem of 
 collaborative decision making over a network of $N$ agents. 
We consider an application in finance where each entry of the decision vector $w$ denotes the amount of investment in a specific 
type of asset. 
Let the $M \times 1$ vector $\bm{p}$ represent the return in investment. 
Each entry of $\bm{p}$ represents the return for a unit investment in the corresponding asset. Let $\overline{p}$ and $R_p$ denote the mean and covariance matrix of $\bm{p}$, respectively. Then, the overall return by the agents for a decision vector $w$ is $\bm{p}^T w$. Note that, with decision $w$, the return $\bm{p}^T w$ is a (scalar) random variable with mean $\overline{p}^T w$ and variance $\var(\bm{p}^T w)=w^T R_p w$,
which are called the expected return and variance of the return in 
classical Markowitz portfolio optimization\cite[p.155]{boyd2004convex},
\cite{markowitz1952portfolio,rubinstein2002markowitz,fitt2009markowitz,elton2010applications}.
These two metrics are often used to characterize the quality of the decision 
$w$: we want
to maximize the expected return while minimizing the variance. However,
solving the problem directly requires all agents to know the global statistics
$\overline{p}$ and $R_p$. What is available in practice are observations that are
collected at the various nodes.
Suppose a subset $\mc{U}$ of the agents observes a sequence of return vectors $\{\bm{u}_{k,i}\}$
with $\E\bm{u}_{k,i}=\overline{p}$. The subscripts $k$ and $i$ denote that the return is
observed by node $k$ at time $i$. 
Then, we can formulate the cost functions for the nodes in set $\mc{U}$
as follows:
	\begin{align}
		\label{Equ:Simu:J_U}
		J_{u,k}(w)	&=	-\E[\bm{u}_{k,i}^T w]
					=	-\overline{p}^Tw
						\\
					&
					k \in \mc{U} \subset \{1,\ldots,N\}
						\nonumber
	\end{align}
We place a negative sign in \eqref{Equ:Simu:J_U} so that
minimizing $J_{u,k}(w)$ is equivalent to maximizing the expected return.
Similarly, suppose there is another subset of nodes, exclusive from $\mc{U}$ and denoted by $\mc{S}$,
which observes a sequence of centered return vectors $\{\bm{s}_{k,i}\}$,
namely, vectors that have the same distribution as
$\bm{p}-\E\bm{p}$ so that $\E[\bm{s}_{k,i}\bm{s}_{k,i}^T]=R_p$. Then,
we can associate with these nodes the cost functions:
	\begin{align}
		\label{Equ:Simu:J_S}
		J_{s,k}(w)	&=	\E\big[|\bm{s}_{k,i}^T w |^2 \big]
					=	w^T R_p w	
						\\
						&k \in \mc{S} \subset \{1,\ldots,N\}
							\nonumber
	\end{align}
Additionally, apart from selecting the decision vector $w$ to maximize the return subject to minimizing its variance, 
the investment strategy $w$ needs to satisfy other constraints
such as: i) the total amount
of investment should be less than a maximum value that
is known only to an agent $k_0 \in \mc{K}$ (e.g., agent $k_0$
is from the funding department who knows how much funding is available), 
ii) the investment on each asset
be nonnegative (known to all agents), and
iii) tax requirements and tax deductions\footnote{For example, 
suppose the first and second entries of the decision
vector $w$ denote the investments on charity assets. 
When the charity investments exceed a certain amount, say $b$,
there would be a tax deduction. We can represent this situation by 
writing $h^T w \ge b$, where $h \defeq [1 \; 1\; 0 \;\cdots \; 0]^T$.}
known to agents in a set $\mc{H}$.
We can then formulate the following constrained multi-objective optimization problem:
	\begin{align}
		\label{Equ:Simu:Optimization:Objectives}
		\min_w			\quad&	\Big\{
										\sum_{ k \in \mc{U} } J_{u,k}(w),\;
										\sum_{ k \in \mc{S} }J_{s,k}(w)
								\Big\}
								\\
		\label{Equ:Simu:Optimization:Constraint1_Sum}
		\mathrm{s.t.}	\quad&	\mathds{1}^T w	\leq		b_0							\\
		\label{Equ:Simu:Optimization:Constraint2_Tax}
						\quad&	h_k^T w			\ge 		b_k,		\quad	k \in \mc{H}\\
		\label{Equ:Simu:Optimization:Constraint3_Nonnegative}
						\quad&	w				\succeq		0
	\end{align}
Using the scalarization technique and barrier function method from Sec. \ref{Sec:Intro},
we convert
\eqref{Equ:Simu:Optimization:Objectives}--\eqref{Equ:Simu:Optimization:Constraint3_Nonnegative}
into the following unconstrained optimization problem (for simplicity, we only 
consider $\pi_1=\cdots=\pi_N=1$):
	\begin{align}
		J^{\mathrm{glob}}(w)		=&\;\sum_{ k \in \mc{U} } 
										\left[
												J_{u,k}(w) 
												+
												\sum_{m=1}^M \phi(-e_m^Tw)
										\right]
										\nonumber\\
										&+ 
										\sum_{ k \in \mc{S} }
										\left[
												J_{s,k}(w)
												+
												\sum_{m=1}^M \phi(-e_m^Tw)
										\right]
										\nonumber\\
									&
										+ 
										\sum_{k \in \mc{H}}
										\left[
												\phi(b_k-h_k^Tw)
												+
												\sum_{m=1}^M \phi(-e_m^Tw) 
										\right]
										\nonumber\\
										&+ 
										\sum_{k \in \mc{K}}
										\left[
												\phi(\mathds{1}^Tw - b_0)
												+
												\sum_{m=1}^M \phi(-e_m^Tw)
										\right]
										\nonumber
	\end{align}
where $\phi(\cdot)$ is a barrier function to penalize the violation of the constraints
--- see \cite{towfic2012icc} for an example,
and the vector $e_m \in \mb{R}^M$ is a basis vector whose entries are all zero except
for a value of one at the $m$th
entry. The term $\sum_{m=1}^M \phi(-e_m^Tw)$ is added to each
cost function to enforce the nonnegativity constraint \eqref{Equ:Simu:Optimization:Constraint3_Nonnegative}, which is assumed
to be known to all agents.
Note that there is a ``division of labor'' over the network: 
the entire set of nodes is divided into four mutually exclusive subsets
$\{1,\ldots,N\} = \mc{U} \cup \mc{S} \cup \mc{H} \cup \mc{K}$, 
and each subset 
collects one type of information
related to the decision. Diffusion adaptation strategies allow the nodes to arrive at
a Pareto-optimal decision in a distributed manner over the network, and each
subset of nodes influences the overall investment strategy.

\begin{figure*}[t!]
   \centerline
   {
   	\subfigure[Topology of the network.]
	{
	   \includegraphics[width=2.4in]{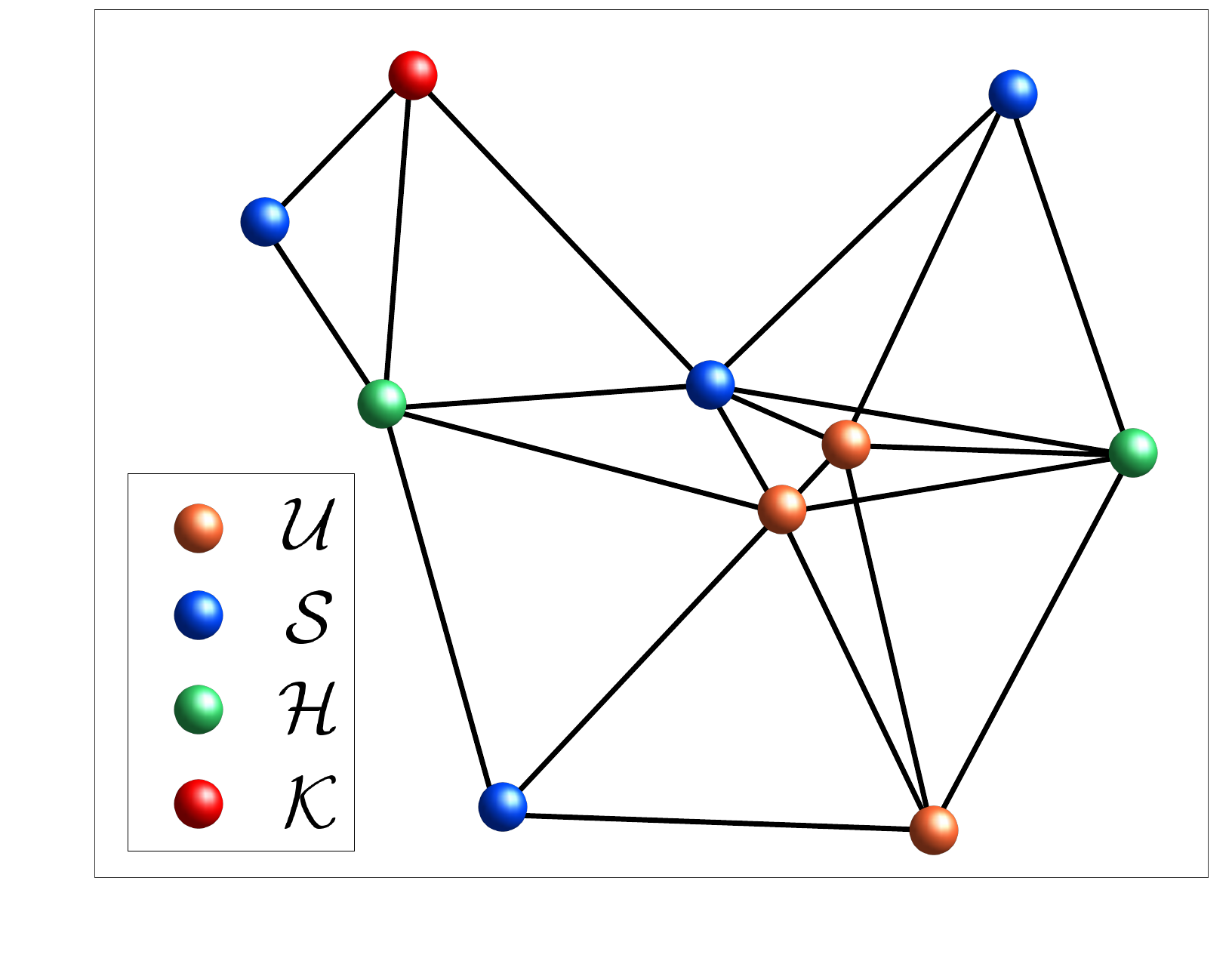}
	   \label{fig:Fig_Simu_Topology}
	}
    \hfil
   	\subfigure[Learning curve ($\mu=10^{-2}$).]
	{	   \includegraphics[width=2.5in]{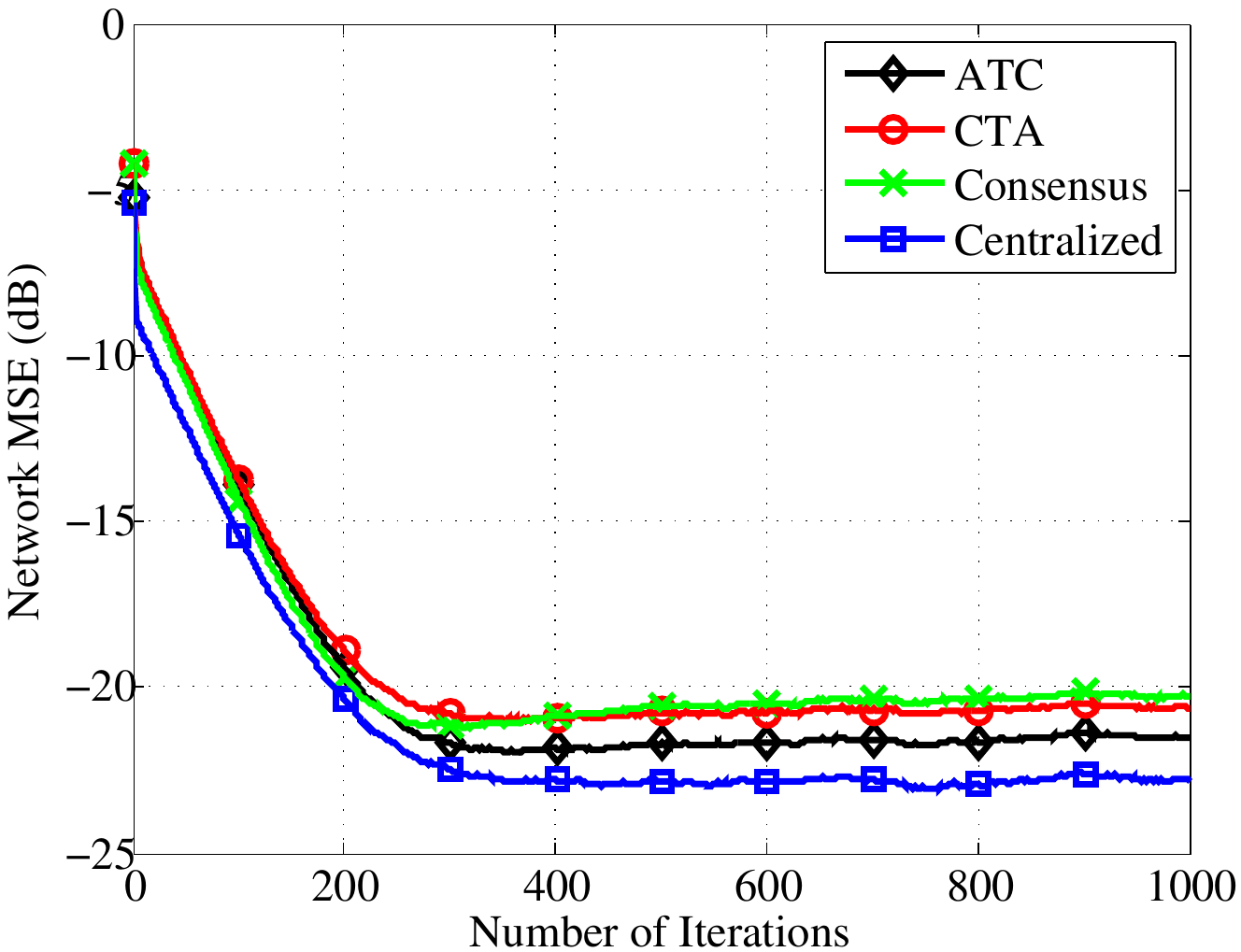}
	   \label{fig:Fig_Simu_LearningCurve}
	  }	
   	}
   \centerline
   {
    \subfigure[MSE for different values of step-sizes.]
	{
	   \includegraphics[width=2.5in]{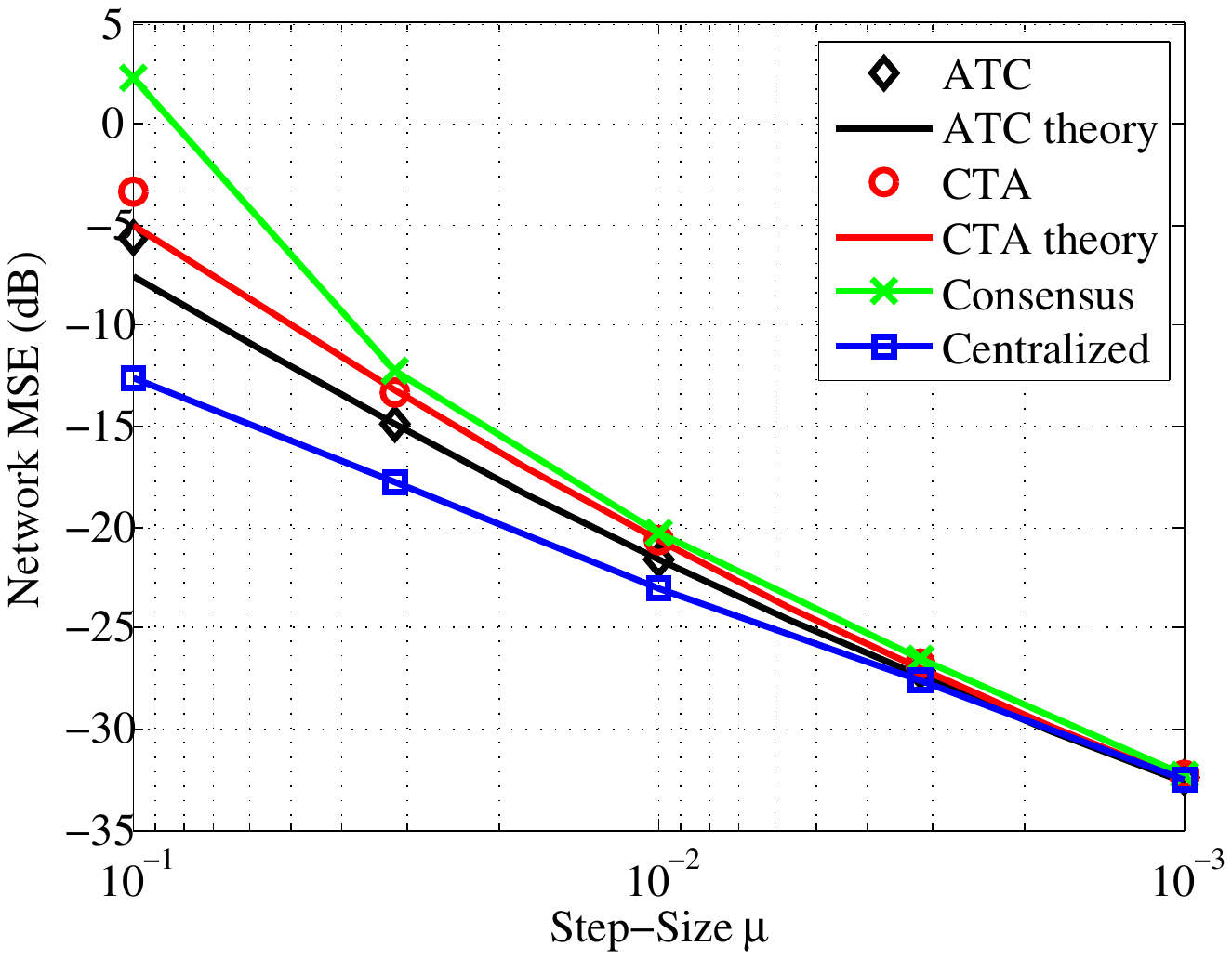}
	   \label{fig:Fig_Simu_MSEvsMu}
   	}
    \hfil	
   	\subfigure[Error of fixed point for different values of step-sizes.]
	{
	   \includegraphics[width=2.5in]{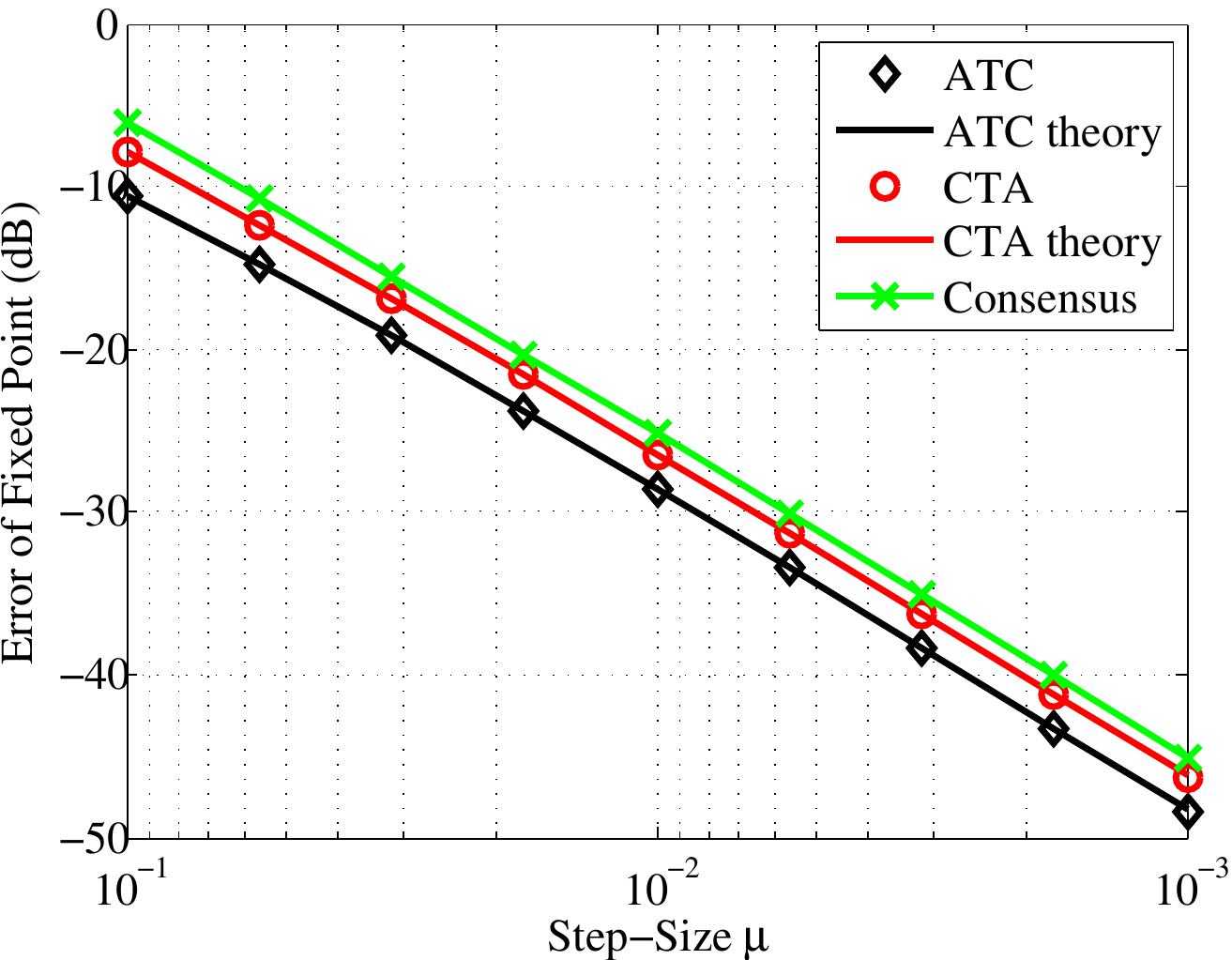}
	   \label{fig:Fig_Simu_FPE}
   	}
   }
   \caption{Simulation results for collaborative decision making.}
   \label{fig:Fig_Simu}
    \vspace{-\baselineskip}
\end{figure*}

In our simulation, we consider a randomly generated connected network topology.
There are a total of $N=10$ nodes in the network, and nodes are assumed connected when
they are close enough geographically. The cardinalities of the subsets $\mc{U}$,
$\mc{S}$, $\mc{H}$ and $\mc{K}$ are set to be $3$, $4$, $2$ and $1$, respectively.
The nodes are partitioned into these four subsets randomly.
The dimension of the decision vector is $M=5$.
The random vectors $\bm{u}_{k,i}$ and $\bm{s}_{k,i}$ are generated according to
the Gaussian distributions $\mc{N}(\mathds{1}, I_M)$ and $\mc{N}(0, I_M)$, respectively. 
We set $b_0=5$ and the parameters $\{h_k, b_k\}$ for $k \in \mc{H}$ to
	\begin{align}
		h_{k_1}	=	\left[
						1	\;	2	\;	\cdots	\;	5		
					\right],
					\quad
		b_{k_1}	=	2
					\\
		h_{k_2}	=	\left[
						5	\;	4	\;	\cdots	\;	1		
					\right],
					\quad
		b_{k_2}	=	3
	\end{align}
where $k_1$ and $k_2$ are the indices of the two nodes in the subset $\mc{H}$.
Furthermore, we use the barrier function given by (15) in \cite{towfic2012icc}
in our simulation with $t=10$, $\rho=0.1$ and $\tau=0.1$. 
We set the combination coefficients $\{a_{lk}\}$ to the Metropolis rule (See Table III in
\cite{Cattivelli10})
 for both ATC and CTA
strategies. The weights $\{c_{lk}\}$ are set to $c_{lk}=1$ for $l=k$ and zero otherwise, i.e.,
there is no exchange of gradient information among neighbors. 
According to Theorem \ref{Thm:Bias_Small_StepSize}, such a choice will always guarantee
condition \eqref{Equ:ConvergenceAnalysis:Thm:ConsensusCondition} so that the bias can 
be made arbitrarily small for small step-sizes. In our simulation, we do not assume the statistics
of $\{\bm{u}_{k,i}\}$ and $\{\bm{s}_{k,i}\}$ are known to the nodes. The only information available is
their realizations and the algorithms have to learn the best decision vector $w$ from them. Therefore, we use the following stochastic gradient vector%
\footnote{For nodes in $\mc{H}$ an $\mc{K}$, the cost functions are
known precisely, so their true gradients are used.}
at each node $k$:
{\begingroup
\renewcommand{\arraystretch}{-5}
	\begin{align}
		\label{Equ:Simu:StochasticGradient}
		\widehat{\nabla_w J_k}(w)
		&=	\begin{cases}
											-\bm{u}_{k,i} + \sum_{m=1}^M \nabla_w \phi(-e_m^Tw)
															&	k	\in 	\mc{U}		\\
											2\bm{s}_{k,i} + \sum_{m=1}^M \nabla_w \phi(-e_m^Tw)
															&	k	\in 	\mc{S}		\\
											\nabla_w \phi(b_k-h_k^T w) + \sum_{m=1}^M \nabla_w \phi(-e_m^Tw)
															&	k	\in 	\mc{H}		\\
											\nabla_w \phi(\mathds{1}^Tw-b_0) + \sum_{m=1}^M \nabla_w \phi(-e_m^Tw)
															&	k	\in 	\mc{K}
										\end{cases}
	\end{align}
\endgroup
}%
To compare the performance with other algorithms,
we also simulate the consensus-based approach from \cite{nedic2009distributed} with the same stochastic gradient%
\footnote{The original algorithm in \cite{nedic2009distributed} does not use
stochastic gradients but the true gradients $\{\nabla_w J_k(w)\}$.} 
as \eqref{Equ:Simu:StochasticGradient}. The algorithm is listed below:
	\begin{align}
		\bm{w}_{k,i}		=	\sum_{l \in \mc{N}_k} a_{lk} \bm{w}_{l,i-1} - \mu \widehat{\nabla_w J_k}(\bm{w}_{k,i-1})
	\end{align}
Furthermore, we also simulate the conventional centralized approach to such
optimization problem, which collects data from all nodes and implements
stochstic gradient descent at the central node:
	\begin{align}
		\bm{w}_{i}		=	\bm{w}_{i-1} - \mu \frac{1}{N} \sum_{k=1}^N \widehat{\nabla_w J_k}(\bm{w}_{i-1})
	\end{align}
where the factor of $1/N$ is used to make the convergence rate the same as the distributed algorithms.

The simulatin results are shown in Fig. \ref{fig:Fig_Simu_Topology}--\ref{fig:Fig_Simu_FPE}. Fig. \ref{fig:Fig_Simu_Topology}
shows the network topology, and Fig. \ref{fig:Fig_Simu_LearningCurve} shows
 the learning curves of different algorithms.
We see that ATC outperforms CTA and CTA outperforms consensus. To further compare the steady-state performance,
we plot the steady-state MSE for different values of step-sizes in Fig. \ref{fig:Fig_Simu_MSEvsMu}. 
We also plot
the theoretical curves from \eqref{Equ:ConvergenceAnalysis:MSE_k}--\eqref{Equ:PerformanceAnalysis:MSE_Network} for ATC
and CTA algorithms. We observe that
all algorithms approach the performance of the centralized solution when the step-sizes are small. However, diffusion algorithms always outperform
the consensus-based strategy; the gap between ATC and consensus algorithm is about $8$ dB when $\mu=0.1$.
We also see that the theoretical curves match the simulated ones well.
Finally, we recall that Theorem \ref{Thm:Bias_Small_StepSize} shows that the error between the fixed point $w_{\infty}$ and
$\mathds{1}\otimes w^o$ can be made arbitrarily small for small step-sizes, and the error $\|w_{\infty}-\mathds{1}\otimes w^o\|^2$
is on the order of $O(\mu^2)$. To illustrate the result, we simulate the algorithms
using true gradients $\{\nabla_w J_k(w)\}$ so that they converge to their fixed point $w_{\infty}$,
and we get different values of $w_{\infty}$ for different step-sizes. The theoretical values for ATC and CTA can
be computed from \eqref{Equ:PerformanceAnalysis:Thm_Convergence:w_tilde_infty_Expr}. The results are shown in Fig. \ref{fig:Fig_Simu_FPE}.
We see that the theory matches simulation, and the power of the fixed point error per node%
\footnote{The power of the fixed point error per node is defined as 
$\frac{1}{N}\|w_{\infty}-\mathds{1}\otimes w^o\|^2 = \frac{1}{N} \sum_{k=1}^N \|w_{k,\infty}-w^o\|^2$.}
decays at $20$dB per decade, which
is $O(\mu^2)$ and is consistent with \eqref{Equ:ConvergenceAnalysis:Thm:SmallBias_mu}. Note that
diffusion algorithms outperform the consensus. Also note from
\eqref{Equ:PerformanceAnalysis:Thm_Convergence:w_tilde_infty_Expr} and 
\eqref{Equ:PerformanceAnalysis:Ew_inf} that the bias and the fixed point error have
the same expression. Therefore, diffusion algorithms have smaller bias than 
consensus (the gap in Fig. \ref{fig:Fig_Simu_FPE} is  as large as $5$dB
between ATC and consensus).

\section{Conclusion}
This paper generalized diffusion adaptation strategies to perform multi-objective
optimization in a distributed manner over a network of nodes.
We use constant step-sizes to endow the network with
continuous learning and adaptation abilities via local interactions.
We analyzed the mean-square-error performance of the diffusion strategy,
and showed that the solution at each node gets arbitrarily close to the same Pareto-optimal
solution for small step-sizes.

\appendices

{
\section{Properties of the Operators}
\label{Appendix:Properties_Operators}
Properties 1-3 are straightforward from the definitions of $T_A(\cdot)$
and $P[\cdot]$. We therefore omit the proof for brevity, and start with property 4.
\\

\noindent(\underline{Property 4: Convexity})\\
We can express each $N\times 1$ block vector $x^{(k)}$ in the form
        $x^{(k)} =   \col\{x_1^{(k)},\ldots,x_N^{(k)}\}$ for $k=1,\ldots,N$.
Then, the convex combination of $x^{(1)},\ldots, x^{(N)}$ can be expressed as
    \begin{align}
        \sum_{k=1}^K a_l \; x^{(k)}  =   \col\Big\{
                                        \sum_{k=1}^K a_l x_1^{(k)},
                                        \ldots,
                                        \sum_{k=1}^K a_l x_N^{(k)}
                                    \Big\}
    \end{align}
According to the definition of the operator $P[\cdot]$, and in view of the convexity of
$\|\cdot\|^2$, we have
{\allowdisplaybreaks
    \begin{align}
        P\Big[\sum_{k=1}^K a_l \; x^{(k)}\Big]
                                    &=  \col\Big\{
                                            \big\|\sum_{k=1}^K a_l x_1^{(k)}\big\|^2,
                                            \ldots,
                                            \big\|\sum_{k=1}^K a_l x_N^{(k)}\big\|^2
                                        \Big\}
                                        \nonumber\\
                                    &\preceq
                                        \col\Big\{
                                            \sum_{k=1}^K a_l\| x_1^{(k)}\|^2,
                                            \ldots,
                                            \sum_{k=1}^K a_l\| x_N^{(k)}\|^2
                                        \Big\}
                                        \nonumber\\
                                    &=   \sum_{k=1}^K a_l \; P[x^{(k)}]
    \end{align}
}%

\noindent(\underline{Property 5: Additivity})\\
By the definition of $P[\cdot]$ and the assumption that $\E\bm{x}_k^T\bm{y}_k=0$
for each $k=1,\ldots,N$, we obtain
{\allowdisplaybreaks
    \begin{align}
        \E P[\bm{x}+\bm{y}]
        						&=      \col\{
                                                \E\|\bm{x}_1+\bm{y}_1\|^2,\;
                                                \ldots,\;
                                                \E\|\bm{x}_N+\bm{y}_N\|^2
                                        \}
                                        \nonumber\\
                                &=      \col\{
                                                \E\|\bm{x}_1\|^2+\E\|\bm{y}_1\|^2,\;
                                                \ldots,\;
                                                \E\|\bm{x}_N\|^2+\E\|\bm{y}_N\|^2
                                        \}
                                        \nonumber\\
                                &=      \E P[\bm{x}] + \E P[\bm{y}]
    \end{align}
}%

\noindent(\underline{Property 6: Variance Relations})\\
We first prove \eqref{Equ:Properties:PX_TA}. From the definition of $T_A(\cdot)$
in \eqref{Equ:def:T_A} and the definition of $P[\cdot]$ in \eqref{Equ:Def:PowerOperator},
we express 
    \begin{align}
        P[T_A(x)]   &=      \col\Big\{
                                \big\| \sum_{l=1}^N a_{l1} x_l \big\|^2,\;
                                \ldots,\;
                                \big\| \sum_{l=1}^N a_{lN} x_l \big\|^2
                            \Big\}
    \end{align}
Since $\|\cdot\|^2$ is a convex function and each sum inside the squared norm operator
is a convex combination of $x_1,\ldots,x_N$ ($A^T$ is right stochastic),
by Jensen's inequality\cite[p.77]{boyd2004convex},
we have
    \begin{align}
        P[T_A(x)]   &\preceq\col\Big\{
                                \sum_{l=1}^N a_{l1}\| x_l \|^2,\;
                                \ldots,\;
                                \sum_{l=1}^N a_{lN}\| x_l \|^2
                            \Big\}
                            \nonumber\\
                    &=      A^T \col\{
                                \| x_1 \|^2,\;
                                \ldots,\;
                                \| x_N \|^2
                            \}
                            \nonumber\\
                    &=       A^T P[x]
    \end{align}
Next, we proceed to prove \eqref{Equ:Properties:PX_TG}. We need to call upon the 
following useful lemmas from \cite[p.24]{poliak1987introduction},
and Lemmas 1--2 in\cite{chen2011TSPdiffopt}, respectively.
    \begin{lemma}[Mean-Value Theorem]
        \label{Lemma:MeanValueTheorem}
        For any twice-differentiable function $f(\cdot)$, it holds that
        	\begin{align}
        		\label{Equ:Appendix:MeanValueTheorem}
        		\nabla f(y)		=	\nabla f(x) + \left[\int_0^1 \nabla^2 f\big(x+t(y-x)\big)dt\right] (y-x)
        	\end{align}
        where $\nabla^2 f(\cdot)$ denotes the Hessian of $f(\cdot)$, and is a symmetric matrix.
        \hfill\qed
    \end{lemma}
    \begin{lemma}[Bounds on the Integral of Hessian]
		\label{Lemma:H_lki_bounds}
		Under Assumption \ref{Assumption:Hessian}, the following bounds
        hold for any vectors ${x}$ and ${y}$:
			\begin{align}
                \label{Equ:ConvergenceAnalysis:H_lki_Bounds}
				&\lambda_{l,\min} I_M	
                \le
                \int_{0}^{1} \nabla_w^2 J_l({x} + t {y}) dt
            	\le	
                \lambda_{l,\max} I_M
                \\
                \label{Equ:PerformanceAnalysis:I_minus_mu_H_bound}
                &
                \left\|
                        I - \mu_k \sum_{l=1}^N c_{lk}
                        \left[\int_{0}^{1} \nabla_w^2 J_l({x} + t {y}) dt\right]
                \right\|
                \le
                \gamma_k
			\end{align}
        where $\|\cdot\|$ denotes the $2-$induced norm, and $\gamma_k$,
        $\sigma_{k,\min}$ and $\sigma_{k,\max}$
        were defined in \eqref{Equ:PerformanceAnalysis:gamma_k_def}--%
        \eqref{Equ:PerformanceAnalysis:Thm_MSS:sigma_min_def}.
        \hfill\qed
	\end{lemma}
\noindent
By the definition of the operator $T_G(\cdot)$ in \eqref{Equ:Def:T_G} and
the expression \eqref{Equ:Appendix:MeanValueTheorem},
we express 
$T_G(x)-T_G(y)$ as
{
\begingroup
\renewcommand*{\arraystretch}{-5}
    \begin{align}
        T_G(x)-T_G(y)=   
                        &  \begin{bmatrix}
                        		\displaystyle
                                \left[
                                    I_M \!-\! \mu_1 \sum_{l=1}^N c_{l1}
                                    \!\! \int_0^1 \!\!
                                    \nabla_w^2 J_l(y_1\!+\!t(x_1\!-\!y_1))dt
                                \right] \!
                                (x_1 \!-\! y_1)
                                \\
                                \vdots\\
                                \displaystyle
                                \left[
                                    I_M \!-\! \mu_N \sum_{l=1}^N c_{lN} 
                                    \!\!\int_0^1 \!\!
                                    \nabla_w^2 J_l(y_N\!+\!t(x_N\!-\!y_N))dt
                                \right] \!
                                (x_N \!-\! y_N)
                            \end{bmatrix}
    \end{align}
\endgroup
}%
Therefore, using \eqref{Equ:PerformanceAnalysis:I_minus_mu_H_bound} and the definition of $P[\cdot]$
in \eqref{Equ:Def:PowerOperator}, we obtain
    \begin{align}
        P[T_G(x)-T_G(y)]
                            &\preceq	\col\left\{
                                        \gamma_1^2 \cdot\|x_1 - y_1\|^2,
                                        \ldots,
                                        \gamma_N^2 \cdot\|x_N - y_N\|^2
                                    \right\}
                                    \nonumber\\
                            &=       \Gamma^2 P[x-y]
    \end{align}

\noindent(\underline{Property 7: Block Maximum Norm})\\
According to the definition of $P[\cdot]$ and the definition of block maximum norm
\cite{chen2011TSPdiffopt}, we have
    \begin{align}
        \|P[x]\|_{\infty}   &=   \big\|\col\{\|x_1\|^2,\;\ldots,\;\|x_N\|^2\}\big\|_{\infty}
        							\nonumber\\
                            &=   \max_{1 \le k \le N} \|x_k\|^2
                            		\nonumber\\
                            &=   \big(\max_{1 \le k \le N} \|x_k\|\big)^2
                            		\nonumber\\
                            &=   \|x\|_{b,\infty}^2
    \end{align}

\noindent(\underline{Property 8: Preservation of Inequality})\\
To prove $Fx \preceq Fy$, it suffices to prove $0 \preceq F(y-x)$. Since
$x \preceq y$, we have $0 \preceq y-x$, i.e., all entries of the vector $y-x$ are
nonnegative. Furthermore, since all entries of the matrix $F$ are nonnegative,
the entries of the vector $F(y-x)$ are all nonnegative, which means $0 \preceq F(y-x)$.

\section{Bias at Small Step-Sizes}
\label{Appendix:BiasSmallStepSizes}
It suffices to show that
	\begin{align}
		\label{Equ:Appendix:SmallBias:lim_w_inf_to_mu_max}
		\lim_{\mu_{\max} \rightarrow 0}	
            \frac{\|
                    \mathds{1}\otimes w^o - w_{\infty}
                \|}{\mu_{\max}}	=	\xi
	\end{align}
where $\xi$ is a constant independent of $\mu_{\max}$.
It is known that any matrix is similar to a Jordan canonical form \cite{laub2005matrix}.
Hence, there exists an invertible matrix $Y$ such that $A_2^T A_1^T	=	Y J Y^{-1}$,
where $J$ is the Jordan canonical form of the matrix $A_2^T A_1^T$,
and the columns of the matrix $Y$ are the corresponding \emph{right principal vectors}
of various degrees
\cite[pp.82--88]{laub2005matrix};
the right principal vector of degree one is the right eigenvector.
Obviously, the matrices $J$ and $Y$ are independent of $\mu_{\max}$.
Using the Kronecker product property\cite[p.140]{laub2005matrix}:
$(A \otimes B)(C \otimes D) = AC \otimes BD$, we obtain
	\begin{align}
		\mathcal{A}_2^T \mathcal{A}_1^T	
					&=	A_2^TA_1^T \otimes I_M			
						\nonumber\\
		\label{Equ:ConvergenceAnalysis:cal_P2P1_JCF}
					&=	(Y\otimes I_M) (J \otimes I_M) (Y^{-1} \otimes I_M)
	\end{align}

{\color{black}
\noindent
Denote $\mu_k = \beta_k \mu_{\max}$, where $\beta_k$ is some positive scalar
such that $0 < \beta_k \le 1$.
}%
Substituting \eqref{Equ:ConvergenceAnalysis:cal_P2P1_JCF} into
\eqref{Equ:PerformanceAnalysis:Thm_Convergence:w_tilde_infty_Expr}, we obtain
	\begin{align}
            \mathds{1}\otimes w^o - w_{\infty}
					&=	\left[
							I_{MN}
							-
							\mathcal{A}_2^T
							\mathcal{A}_1^T
							+
							\mathcal{A}_2^T \mathcal{M}{\mathcal{R}}_{\infty}
							\mathcal{A}_1^T
						\right]^{-1}
						\mathcal{A}_2^T \mathcal{M}\mathcal{C}^T g^o
						\nonumber\\
		\label{Equ:ConvergenceAnalysis:w_infty_SmallStepSize_intermediate}
					&=	(Y\otimes I_M)
						\left[
							I_{MN} - J \otimes I_M + \mu_{\max} E
						\right]^{-1}
						 (Y^{-1} \otimes I_M)
						\mathcal{A}_2^T \mathcal{M}\mathcal{C}^Tg^o
	\end{align}
where 
		\begin{align}
			\label{Equ:ConvergenceAnalysis:E}
			E	&=	(Y^{-1} \otimes I_M)
					\mathcal{A}_2^T \mathcal{M}_0 \mathcal{R}_\infty \mathcal{A}_1^T
					(Y\otimes I_M)
					\\
		\label{Equ:Appendix:M_0_Omega_0_def}
		\mathcal{M}_0     &\defeq  \mathcal{M}/\mu_{\max}	
                            =	  \underbrace{
												\mathrm{diag}\{\beta_1,\ldots,\beta_N\}
											}_{\defeq\Omega_0}
											\otimes
											I_M
	\end{align}
By \eqref{Equ:PerformanceAnalysis:ConditionCombinationWeights},
the matrix $A_2^TA_1^T$ is right-stochastic, and
since $A_2^TA_1^T$ is regular, it will have an eigenvalue of one that has multiplicity one and
is strictly greater than all other eigenvalues\cite{horn1990matrix}.
Furthermore, the corresponding left and right eigenvectors are $\theta^T$ and $\mathds{1}$,
with $\theta^T \succ 0$ (all entries of the row vector $\theta^T$ are real positive numbers).
For this reason, we can partition $J$, $Y^{-1}$ and $Y$ in the following block forms:
	\begin{align}
		\label{Equ:ConvergenceAnalysis:J_Y_Yinv}
		J		=	\diag\{1	,\;J_0\},
					\quad
		Y^{-1}	=	\col\left\{
						\frac{\theta^T}{\theta^T \mathds{1}},\;
						Y_R
						\right\},
					\quad
		Y		=	\left[
						\mathds{1}	\;	Y_L
					\right]
	\end{align}
where $J_0$ is an $(N-1)\times(N-1)$ matrix that contains the Jordan blocks of
eigenvalues strictly within unit circle, i.e., $\rho(J_0)<1$.
{\color{black}
The first row of the matrix $Y^{-1}$ in \eqref{Equ:ConvergenceAnalysis:J_Y_Yinv}
is normalized by $\theta^T \mathds{1}$ so that $Y^{-1}Y=I$. (Note that $Y^{-1}Y=I$ requires
the product of the first row of $Y^{-1}$ and the first column of $Y$ to be one: $\frac{\theta^T}{\theta^T \mathds{1}} \mathds{1}=1$.)
}%
Substituting these partitionings into
\eqref{Equ:ConvergenceAnalysis:E}, we can express $E$ as
	\begin{align}
		\label{Equ:ConvergenceAnalysis:E2}
			E		&=	\begin{bmatrix}
						E_{11}	&	E_{12}		\\
						E_{21}	&	E_{22}
					\end{bmatrix}
	\end{align}
where
	\begin{align}
		\label{Equ:ConvergenceAnalysis:e_0}
		E_{11}	&\defeq	\big(\frac{\theta^T}{\theta^T \mathds{1}} \otimes I_M\big)
						\mathcal{A}_2^T \mathcal{M}_0 \mathcal{R}_\infty \mathcal{A}_1^T
						(\mathds{1} \otimes I_M)
						\\
		E_{12}	&\defeq	\big(\frac{\theta^T}{\theta^T \mathds{1}} \otimes I_M\big)
						\mathcal{A}_2^T \mathcal{M}_0 \mathcal{R}_\infty \mathcal{A}_1^T
						(Y_L \otimes I_M)
						\\
		E_{21}	&\defeq	(Y_R \otimes I_M)
						\mathcal{A}_2^T \mathcal{M}_0 \mathcal{R}_\infty \mathcal{A}_1^T
						(\mathds{1} \otimes I_M)
						\\
		E_{22}	&\defeq	(Y_R \otimes I_M)
						\mathcal{A}_2^T \mathcal{M}_0 \mathcal{R}_\infty \mathcal{A}_1^T
						(Y_L \otimes I_M)
	\end{align}
Observe that the matrices $E_{11}$, $E_{12}$, $E_{21}$ and $E_{22}$ are independent of $\mu_{\max}$.
Substituting \eqref{Equ:ConvergenceAnalysis:J_Y_Yinv} and \eqref{Equ:ConvergenceAnalysis:E2} into \eqref{Equ:ConvergenceAnalysis:w_infty_SmallStepSize_intermediate}, we obtain
	\begin{align}
		\label{Equ:Appendix:Ew_inf_intermediate0}
            \mathds{1}\otimes w^o \!-\! w_{\infty}
                        &=			(Y\otimes I_M)
									\begin{bmatrix}
										\mu_{\max} E_{11}		&		\mu_{\max} E_{12}	
											\\
										\mu_{\max} E_{21}		&		I \!-\! J_0 \otimes I_M
																\!+\! \mu_{\max} E_{22}
									\end{bmatrix}^{-1}
									\begin{bmatrix}
										\frac{1}{\theta^T \mathds{1}}
										(\theta^T \otimes I_M)
										\mathcal{A}_2^T \mathcal{M}\mathcal{C}^Tg^o	\\
										(Y_R \otimes I_M)
										\mathcal{A}_2^T \mathcal{M}\mathcal{C}^Tg^o
									\end{bmatrix}
	\end{align}
Let us denote
	\begin{align}
		\label{Equ:ConvergenceAnalysis:G}
		\begin{bmatrix}
			G_{11}			&		G_{12}			\\
			G_{21}			&		G_{22}
		\end{bmatrix}
		\defeq
		\begin{bmatrix}
			\mu_{\max} E_{11}		&		\!\mu_{\max} E_{12}		\\
			\mu_{\max} E_{21}		&		\!I \!-\! J_0 \otimes I_M
									\!+\! \mu_{\max} E_{22}
		\end{bmatrix}^{-1}
	\end{align}
Furthermore, recalling that $w^o$ is the minimizer of the global cost function
\eqref{Equ:ProblemFormulation:J_glob_original}, we have
	\begin{align}
		\sum_{l=1}^N \nabla_w J_l(w^o)	=	0
		\quad\Leftrightarrow\quad
		(\mathds{1}^T \otimes I_M) \;g^o = 0
	\end{align}
which, together with condition \eqref{Equ:ConvergenceAnalysis:Thm:ConsensusCondition}, implies that
	\begin{align}
		(\theta^T \otimes I_M)\mathcal{A}_2^T \mathcal{M}\mathcal{C}^Tg^o
				=&\;	(\theta^T A_2^T \Omega C^T \otimes I_M)g^o
					\nonumber\\
				=&\;	c_0 (\mathds{1}^T \otimes I_M) g^o
					\nonumber\\
		\label{Equ:Appendix:OptimalCondition_Stationarity}
				=&\;	0
	\end{align}
where we also used the facts that
$\mc{A}_2^T=A_2^T \otimes I_M$, $\mc{C}^T=C^T\otimes I_M$, $\mc{M}=\Omega\otimes I_M$
and the Kronecker product property: $(A \otimes B) (C \otimes D)$.
Substituting \eqref{Equ:ConvergenceAnalysis:G} and
\eqref{Equ:Appendix:OptimalCondition_Stationarity} into
\eqref{Equ:Appendix:Ew_inf_intermediate0} and using \eqref{Equ:Appendix:M_0_Omega_0_def}
lead to
	\begin{align}
		\label{Equ:ConvergenceAnalysis:w_infty_intermediate1}
            \mathds{1}\otimes w^o - w_{\infty}
                        	&=	\mu_{\max}
								\cdot
								(Y\otimes I_M)\!
								\begin{bmatrix}
									G_{12}
									\\
									G_{22}
								\end{bmatrix}
								(Y_R A_2^T \Omega_0 C^T \otimes I_M) g^o
	\end{align}
To proceed with analysis, we need to evaluate $G_{12}$ and $G_{22}$. We call upon the
relation from \cite[pp.48]{laub2005matrix}:
	\begin{align}
		\label{Equ:ConvergenceAnalysis:BlockMatrixInversion}
		\begin{bmatrix}
			P	&	Q	\\
			U	&	V
		\end{bmatrix}^{-1}
		=
		\begin{bmatrix}		
			P^{-1} + P^{-1} Q S U P^{-1}		&	-P^{-1} Q S	\\
			-S U P^{-1}					&	S		
		\end{bmatrix}
	\end{align}
where $S = (V-U P^{-1} Q)^{-1}$. To apply the above relation to \eqref{Equ:ConvergenceAnalysis:G},
we first need to verify that $E_{11}$ is invertible. By definition
\eqref{Equ:ConvergenceAnalysis:e_0},
	\begin{align}			
		E_{11}		=&\;	\big(\frac{\theta^T}{\theta^T \mathds{1}} A_2^T \Omega_0 \otimes I_M\big)
						\mathcal{R}_\infty
						(A_1^T \mathds{1} \otimes I_M)
						\nonumber\\
					=&\;	(z^T \otimes I_M) \mathcal{R}_\infty (\mathds{1} \otimes I_M)
						\nonumber\\
		\label{Equ:Appendix:SmallBias:E11_intermediate}
					=&\;	\sum_{k=1}^N z_{k}
						\sum_{l=1}^N c_{lk}
						H_{lk,\infty}
	\end{align}
where $z_k$ denotes the $k$th entry of the vector $z \defeq \Omega_0 A_2 \theta/\theta^T\mathds{1}$
(note that all entries of $z$ are non-negative, i.e., $z_k \ge 0$).
{\color{black}
Recall from \eqref{Equ:ConvergenceAnalysis:H_kinf} that $H_{lk,\infty}$ is a
symmetric positive semi-definite matrix. Moreover, since $z_k$ and $c_{lk}$ are
nonnegative, we can conclude from \eqref{Equ:Appendix:SmallBias:E11_intermediate} that
$E_{11}$ is a symmetric positive semi-definite matrix. Next, we show that
$E_{11}$ is actually strictly positive definite.
Applying \eqref{Equ:ConvergenceAnalysis:H_lki_Bounds} to the expression
of $H_{lk,\infty}$ in \eqref{Equ:ConvergenceAnalysis:H_kinf}, we obtain
        $H_{lk,\infty}   \ge     \lambda_{l,\min} I_M$.
Substituting 
into
\eqref{Equ:Appendix:SmallBias:E11_intermediate} gives:
}%
	\begin{align}
		E_{11}		\ge&\;	
						\left[
							\sum_{k=1}^N z_{k}
							\sum_{l=1}^N c_{lk}
							\lambda_{l,\min}
						\right]
						\cdot
						I_M
						\nonumber\\
					\ge&\;	
						\Big(\sum_{k=1}^N z_k\Big)
						\min_{1 \le k \le N}
						\Big\{
							\sum_{l=1}^N c_{lk}
							\lambda_{l,\min}
						\Big\}
						\cdot
						I_M
						\nonumber\\
					=&\;
						\frac{\mathds{1}^T\Omega_0 A_2 \theta}{\theta^T\mathds{1}}
						\cdot
						\min_{1 \le k \le N}
						\Big\{
							\sum_{l=1}^N c_{lk}
							\lambda_{l,\min}
						\Big\}
						\cdot
						I_M
	\end{align}
Noting that the matrices $\Omega_0$ and $A_0$ have nonnegative entries with
some entries being positive, and that all entries of $\theta$ are positive,
we have $({\mathds{1}^T\Omega_0 A_2 \theta})/({\theta^T\mathds{1}})>0$.
Furthermore, by Assumption \ref{Assumption:Hessian}, we know
$\sum_{l=1}^N c_{lk} \lambda_{l,\min}>0$ for each $k=1,\ldots,N$.
Therefore, we conclude that $E_{11}>0$ and is invertible.
Applying \eqref{Equ:ConvergenceAnalysis:BlockMatrixInversion} to
\eqref{Equ:ConvergenceAnalysis:G}, we get
		\begin{align}
			G_{12}	&=	-E_{11}^{-1} E_{12}
						G_{22}
						\\
			\label{Equ:Appendix:SmallBias:G_22_Expr}
			G_{22}	&=	\left[
							I \!-\! 
							J_0 \otimes I_M \!+\! 
							\mu_{\max} (E_{22} \!-\! E_{21} E_{11}^{-1} E_{12}^T)
						\right]^{-1}
		\end{align}
Substituting \eqref{Equ:Appendix:SmallBias:G_22_Expr}
into \eqref{Equ:ConvergenceAnalysis:w_infty_intermediate1} leads to
	\begin{align}
		\label{Equ:Appendix:SmallBias:w_infty_final}
            \mathds{1}\otimes w^o - w_{\infty}
                            =&\;	\mu_{\max}
								\cdot
								(Y\otimes I_M)\!
								\begin{bmatrix}
									-E_{11}^{-1} E_{12}
									\\
									I
								\end{bmatrix}
								G_{22}
								(Y_R A_2^T \Omega_0 C^T \otimes I_M) g^o
	\end{align}
Substituting expression \eqref{Equ:Appendix:SmallBias:w_infty_final} into
the left-hand side of \eqref{Equ:Appendix:SmallBias:lim_w_inf_to_mu_max}, we get
	\begin{align}
		\label{Equ:Appendix:SmallBias:lim_Ew_to_mu_max_intermediate1}
		\lim_{\mu_{\max} \rightarrow 0} \frac{\|\mathds{1}\otimes w^o - w_{\infty}\|}{\mu_{\max}}	
			&=			\!\!\!\lim_{\mu_{\max} \rightarrow 0}
						\bigg\|
							(Y\otimes I_M)\!
							\begin{bmatrix}
								-E_{11}^{-1} E_{12}
								\\
								I
							\end{bmatrix}
							G_{22}
							(Y_R A_2^T \Omega_0 C^T \otimes I_M) g^o
						\bigg\|
	\end{align}
Observe that the only term on the right-hand side of
\eqref{Equ:Appendix:SmallBias:lim_Ew_to_mu_max_intermediate1} that depends on $\mu_{\max}$
is $G_{22}$. From its expression \eqref{Equ:Appendix:SmallBias:G_22_Expr},
we observe that, as $\mu_{\max} \rightarrow 0$, the matrix $G_{22}$ tends to
$(I-J_0\otimes I_M)^{-1}$, which is independent of $\mu_{\max}$.
Therefore, the limit on the right-hand side of \eqref{Equ:Appendix:SmallBias:lim_Ew_to_mu_max_intermediate1}
is independent of $\mu_{\max}$.

}

\bibliographystyle{IEEEbib}
\bibliography{DistOpt}

\end{document}